\documentclass[11pt,letterpaper]{article}
\usepackage[utf8]{inputenc}
\usepackage{fullpage,times}
\usepackage[numbers]{natbib}

\usepackage{amsmath,amsfonts,amsthm,amssymb,multirow}
\usepackage{cancel} 
\usepackage{mathtools}

\usepackage{lineno}

\usepackage{algorithm}
\usepackage[noend]{algpseudocode}

\usepackage{parskip} 

\newcommand\thetarel{\operatorname{\theta}}

\newcommand\thetahrel{\operatorname{\hat\theta}}

\newcommand\suchthat{\operatorname{|}}

\newcommand\transthetarel{\operatorname{\hat\theta}}

\newcommand\taurel{\operatorname{\tau}}

\newcommand\factorrel{\operatorname{(\theta\cup\tau)^*}}

\newcommand\edgediff[5]{\left[d_{#1}({#2}, {#4}) - d_{#1}({#2},{#5})\right] - \left[d_{#1}({#3},{#4}) - d_{#1}({#3},{#5})\right]}

\usepackage{braket,amsfonts}

\usepackage{array}

\usepackage[caption=false]{subfig}
\usepackage{pgfplots}

\newcounter{names}
\counterwithin{names}{section}
\newtheorem{claim}[names]{Claim}
\newtheorem{thm}[names]{Theorem}

\newtheorem{lem}[names]{Lemma}
\newtheorem{definition}[names]{Definition}

\usepackage{graphicx,epstopdf}

\usepackage{amsopn}

\usepackage{xspace}
\usepackage{bold-extra}
\usepackage[most]{tcolorbox}

\colorlet{texcscolor}{blue!50!black}
\colorlet{texemcolor}{red!70!black}
\colorlet{texpreamble}{red!70!black}
\colorlet{codebackground}{black!25!white!25}

\lstdefinestyle{siamlatex}{%
  style=tcblatex,
  texcsstyle=*\color{texcscolor},
  texcsstyle=[2]\color{texemcolor},
  keywordstyle=[2]\color{texemcolor},
  moretexcs={cref,Cref,maketitle,mathcal,text,headers,email,url},
}

\tcbset{%
  colframe=black!75!white!75,
  coltitle=white,
  colback=codebackground, 
  colbacklower=white, 
  fonttitle=\bfseries,
  arc=0pt,outer arc=0pt,
  top=1pt,bottom=1pt,left=1mm,right=1mm,middle=1mm,boxsep=1mm,
  leftrule=0.3mm,rightrule=0.3mm,toprule=0.3mm,bottomrule=0.3mm,
  listing options={style=siamlatex}
}

\newtcblisting[use counter=example]{example}[2][]{%
  title={Example~\thetcbcounter: #2},#1}

\newtcbinputlisting[use counter=example]{\examplefile}[3][]{%
  title={Example~\thetcbcounter: #2},listing file={#3},#1}

\DeclareTotalTCBox{\code}{ v O{} }
{ 
  fontupper=\ttfamily\color{black},
  nobeforeafter,
  tcbox raise base,
  colback=codebackground,colframe=white,
  top=0pt,bottom=0pt,left=0mm,right=0mm,
  leftrule=0pt,rightrule=0pt,toprule=0mm,bottomrule=0mm,
  boxsep=0.5mm,
  #2}{#1}

\patchcmd\newpage{\vfil}{}{}{}
\title{Factorization and pseudofactorization of weighted graphs}
\usepackage{authblk}

\author[1]{Kristin Sheridan \thanks{To whom correspondence should be addressed: kristin.sheridan3@gmail.com}\thanks{Research performed while at the Department of Electrical Engineering and Computer Science, Massachusetts Institute of Technology, Cambridge, MA}}
\author[2]{Joseph Berleant}
\author[2]{Mark Bathe}
\author[3]{Anne Condon}
\author[4]{Virginia Vassilevska Williams}
\affil[1]{Department of of Computer Science, University of Texas, Austin, TX}
\affil[2]{Department of Biological Engineering, Massachusetts Institute of Technology, Cambridge, MA}
\affil[3]{Department of Computer Science, University of British Columbia, Vancouver, Canada}
\affil[4]{Computer Science and Artificial Intelligence Laboratory, Massachusetts Institute of Technology, Cambridge, MA}

\begin{document}
\date{}
\maketitle

\begin{abstract}

For unweighted graphs, finding isometric embeddings is closely related to decompositions of $G$ into Cartesian products of smaller graphs. When $G$ is isomorphic to a Cartesian graph product, we call the factors of this product a factorization of $G$. When $G$ is isomorphic to an isometric subgraph of a Cartesian graph product, we call those factors a pseudofactorization of $G$. Prior work has shown that an unweighted graph's pseudofactorization can be used to generate a canonical isometric embedding into a product of the smallest possible pseudofactors. However, for arbitrary weighted graphs, which represent a richer variety of metric spaces, methods for finding isometric embeddings or determining their existence remain elusive, and indeed pseudofactorization and factorization have not previously been extended to this context. In this work, we address the problem of finding the factorization and pseudofactorization of a weighted graph $G$, where $G$ satisfies the property that every edge constitutes a shortest path between its endpoints. We term such graphs minimal graphs, noting that every graph can be made minimal by removing edges not affecting its path metric. We generalize pseudofactorization and factorization to minimal graphs and develop new proof techniques that extend the previously proposed algorithms due to Graham and Winkler [Graham and Winkler, '85] and Feder [Feder, '92] for pseudofactorization and factorization of unweighted graphs. We show that any $m$-edge, $n$-vertex graph with positive integer edge weights can be factored in $O(m^2)$ time, plus the time to find all pairs shortest paths (APSP) distances in a weighted graph, resulting in an overall running time of $O(m^2+n^2\log\log n)$ time. We also show that a pseudofactorization for such a graph can be computed in $O(mn)$ time, plus the time to solve APSP, resulting in an $O(mn+n^2\log\log n)$ running time.
\end{abstract}
\newpage






%
%
%

\section{Introduction}


The task of finding \emph{isometric embeddings}, or mappings of the vertex set of one graph to another while preserving pairwise distances between vertices, is widely applicable but unsolved for general weighted graphs. 
One of the most important applications is in molecular engineering. Attempting
 to design biomolecules with the same control as seen in natural biological systems, molecular engineers may focus on designing sets of DNA strands with pre-specified binding strengths, as these binding strengths can be essential to the emergent behavior of a network of interacting molecules \cite{short2012promiscuous, zhu2010noncognate, antebi2017combinatorial, malinauskas2014extracellular}. Under certain conditions, the binding strength between pairs of DNA strands can be approximated by the distance between pairs of vertices in a hypercube graph, and so the DNA strand-design problem reduces to the task of finding a mapping between a graph whose pairwise vertex distances correspond to desired binding strengths and the hypercube graph whose distances correspond to the actual binding strengths between pairs of DNA strands (Figure \ref{fig:dna-comparison}). This strand design problem could be applied for DNA data storage \cite{baum1995building, neel2006semantic}, DNA logic circuits \cite{qian2011scaling, zhang2011dynamic}, and DNA neural networks \cite{qian2011neural, cherry2018scaling}.
 
 Isometric embeddings may also be applied in communications networks by embedding the connectivity graph of a network into a Hamming graph, or product of complete graphs, which allows shortest paths between nodes to be computed using only local connectivity information \cite{graham1971addressing}; in linguistics as a method of representing the similarities between various linguistic objects \cite{firsov1965isometric}; and in coding theory for the design of certain error-checking codes \cite{kautz1958unit}.

\begin{figure}
    \centering
    \includegraphics[width=\linewidth]{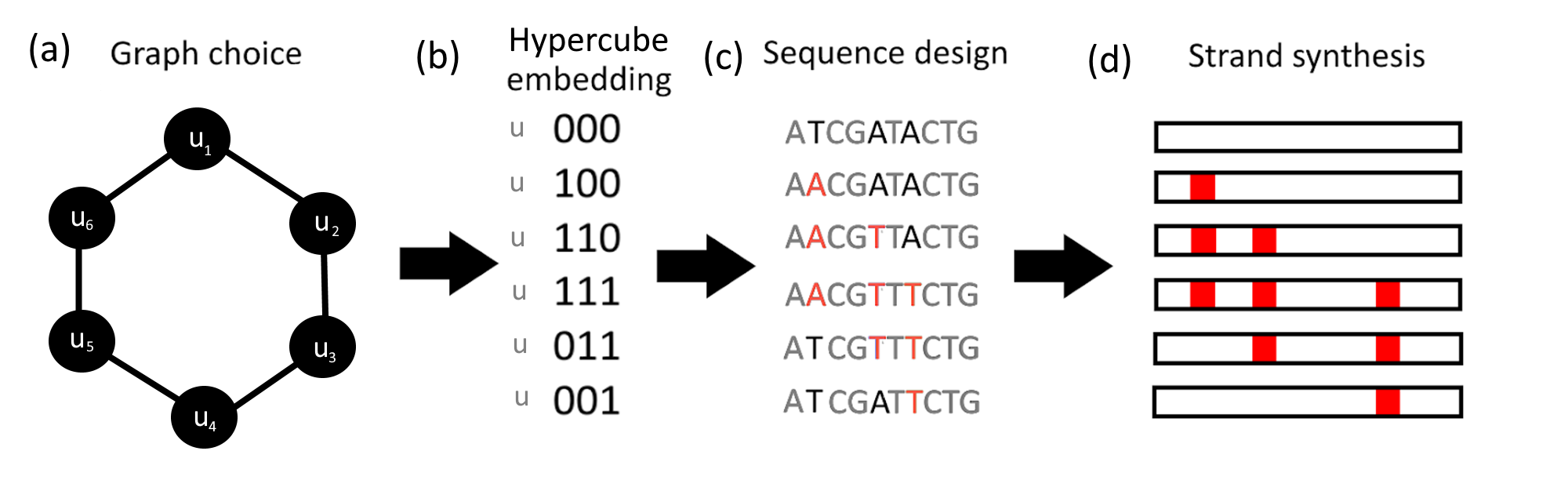}
    \caption[Schematic of the relationship between hypercube embeddability and DNA sequence design]{A schematic showing how hypercube embeddings can be used to construct a set of DNA sequences with a particular relationship. a) An unweighted graph representing the desired binding strength relationships between DNA strands. b) A hypercube embedding (equivalent to an assignment of binary strings) of this graph. c-d) Each binary string in the hypercube embedding corresponds to a DNA sequence, by associating 0 and 1 with point mutations within a sequence. When mutations are not adjacent and the pre- and post-mutation bases are chosen carefully, the binding strength between DNA strands is approximately determined by the number of mutations between one strand and the complement of the other. This relationship can be captured in a hypercube or Hamming graph.}
    \label{fig:dna-comparison}
\end{figure}

In general, many graphs will not have isometric embeddings into a particular destination graph, and the problem of efficiently finding isometric embeddings has been solved for only certain classes of graphs. Prior work has addressed this task for unweighted graphs into hypercubes \cite{djokovic1973distance}, Hamming graphs \cite{winkler1984isometric, wilkeit1990isometric}, and arbitrary Cartesian graph products \cite{graham1985isometric, feder}. These works on unweighted graphs related the isometric embedding problem to representations of a graph either as isomorphic to a Cartesian product of graphs, which we call a \emph{factorization}, or as isomorphic to an isometric subgraph of a Cartesian product of graphs, which we call a \emph{pseudofactorization}. In this work, we extend the concepts of factorization and pseudofactorization to weighted graphs, a task that to the best of our knowledge has not been addressed before. We also note that due to the work in Berleant \textit{et al} \cite{joseph_citation}, this has implications for finding hypercube and Hamming embeddings of certain kinds of graphs.

Except for Section \ref{sec:factorization}, which applies to all weighted graphs, this paper focuses on weighted graphs for which every edge is a shortest path between its endpoints. We call such graphs \emph{minimal graphs}. 
Minimal graphs are a natural subset of weighted graphs, as when we care about shortest paths, edges whose weights are larger than the distance between their endpoints are superfluous.

With this in mind, while we mostly focus on minimal graphs, many aspects of our results are also applicable to arbitrary weighted graphs, because any weighted graph may be made minimal simply by removing any edges that do not affect its path metric. Our results also apply in some cases to arbitrary finite metric spaces, because any finite metric space has a corresponding minimal graph generated by taking a weighted complete graph and removing all extra edges. Other weighted graphs may also be constructed for a given finite metric space, and the isometric embeddings described by our methods will in general depend on which weighted graph representation is used \cite{DezaLaurentText_20.4AdditionalNotes}. 

\subsection{Other work}

\subsubsection{Factorization and pseudofactorization of unweighted graphs}

The \emph{Cartesian graph product} (defined formally in Section \ref{sec:preliminaries})
combines $k\geq 1$ graphs called {\em factors}, so that every vertex in the product graph is a $k$-tuple of vertices, one from each of the factors, and each edge of the product graph corresponds to a single edge from a single one of the factors (see Figure \ref{fig:basic-cart-prod}).

\begin{figure}
    \centering
    \includegraphics[width=.7\linewidth]{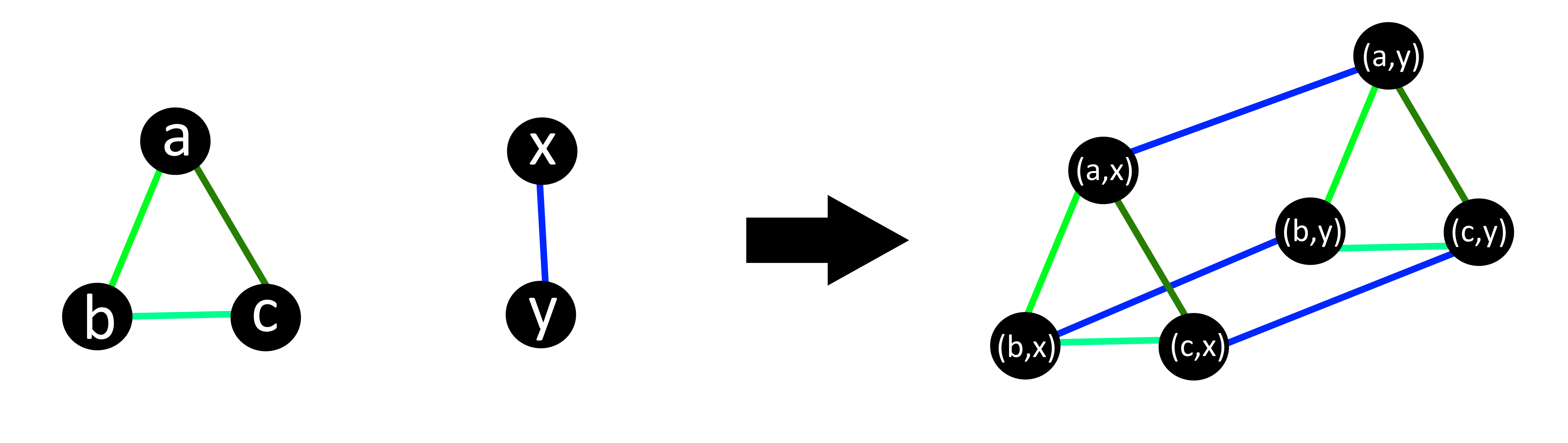}
    \caption[An example of a Cartesian product]{The graph on the right is the Cartesian product of the graphs on the left. The parent edge of a given edge in the product graph is the edge in a factor that has the same color as it. For example, edges $(a,x)(a,y)$, $(b,x)(b,y)$, and $(c,x)(c,y)$ in the product graph all have edge $xy$ in the second factor as their parent edge.}
    \label{fig:basic-cart-prod}
\end{figure}

The problem of finding a representation of a given graph as a Cartesian product of factor graphs is of central importance to isometric embeddings because of the property that every distance in the product graph may be decomposed as a sum of distances in the factor graphs.
This problem may take two forms: the factorization problem and the pseudofactorization problem (Figure \ref{fig:factor}).
Factorization shows that the given graph is isomorphic to a Cartesian product. Pseudofactorization shows that the graph is isomorphic to an isometric subgraph of a Cartesian product.
For unweighted graphs, both problems are well studied as we see below. 

\begin{figure}

\centering
\includegraphics[width=.4\linewidth]{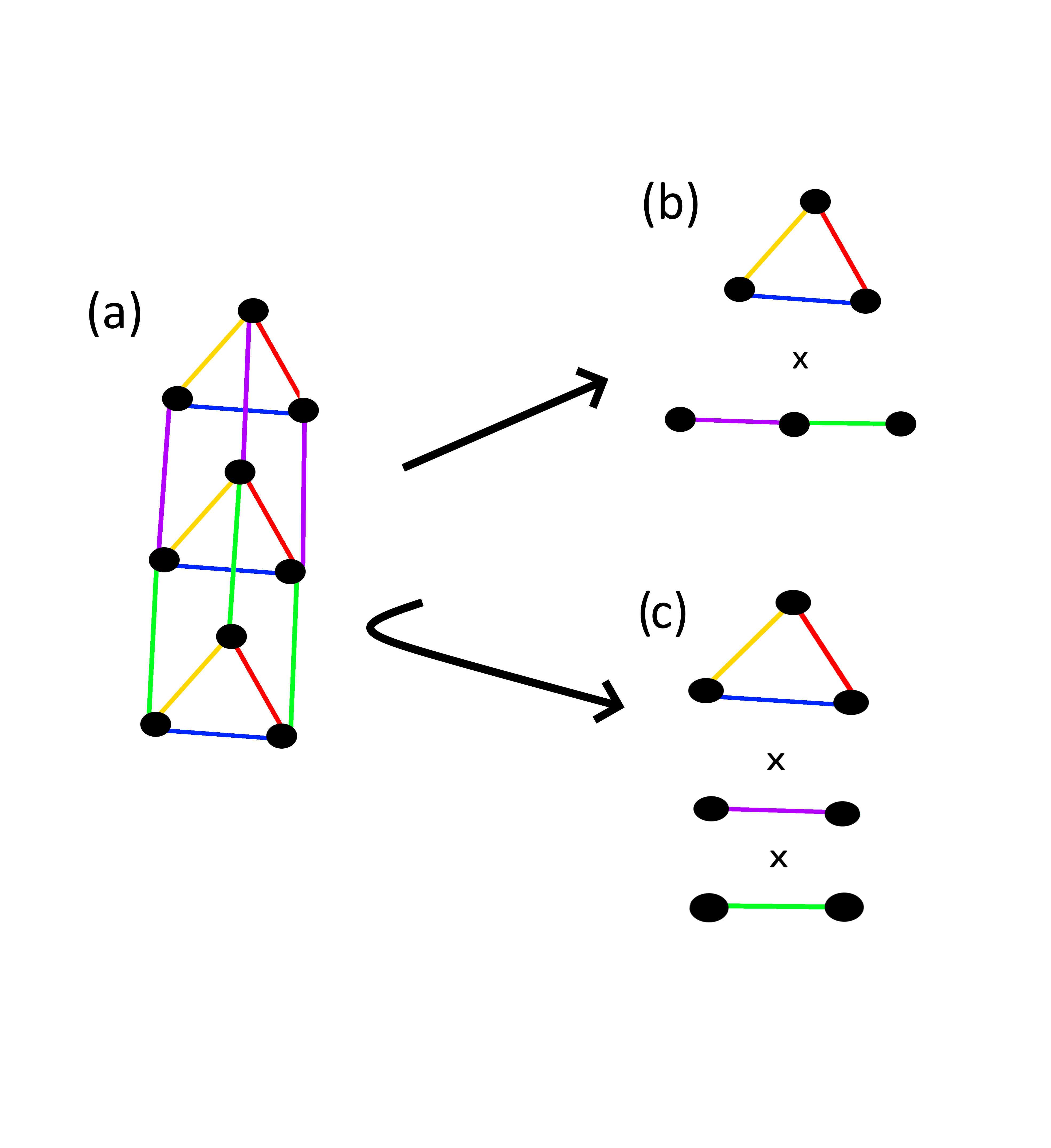}

\caption[An example of factorization and pseudofactorization]{ Subfigure (a) shows a non-prime, non-irreducible graph. Subfigure (b) shows the prime factorization of the graph in (a). Subfigure (c) shows the canonical pseudofactorization of the graph in (a).}
\label{fig:factor}
\end{figure}

Given a graph, a factorization is a set of factor graphs whose Cartesian product is isomorphic to the given graph. A \emph{prime} graph is one whose factorizations always include itself. Graham and Winkler~ \cite{graham1985isometric} showed that every connected, unweighted graph has a unique factorization into prime graphs, or \emph{prime factorization}. Feder \cite{feder} showed that prime factorization for any unweighted $m$-edge, $n$-vertex graph can be found in $O(mn)$ time. Imrich and Klav\v{z}ar \cite{ImrichHammingRecognize} showed that deciding if the prime factorization of an unweighted $m$-edge graph consists entirely of complete graphs can be done in $O(m)$ time. More recently, Imrich and Peterin \cite{imrichLinearFactoring} showed that finding the prime factorization of an arbitrary unweighted $m$-edge graph can also be done in $O(m)$ time. However, for a disconnected graph, this problem is at least as hard as graph isomorphism, which is not known to be in P. As a simple example, consider two connected graphs $G$, $H$, and note that the disjoint union $G\cup H$ has the graph of two disconnected nodes as a factor if and only if $G$ and $H$ are isomorphic (see \cite{DezaLaurentText3}). In fact, the prime factorization of a disconnected graph is no longer unique \cite{winkler1987metric}.

Pseudofactorization generalizes the concept of factorization, by only requiring that the input graph be isometrically embeddable into the Cartesian product of a set of graphs. Clearly, any factorization is also a pseudofactorization; however, the converse is not true (e.g., see Figure \ref{fig:factor}). The analog of a prime graph in the context of pseudofactorization is an \emph{irreducible} graph. For unweighted graphs, Graham and Winkler's prior study of pseudofactorization \cite{graham1985isometric} showed that each connected, unweighted graph has a unique pseudofactorization into irreducible graphs, its \emph{canonical pseudofactorization}.\footnote{Notably, Graham and Winkler \cite{graham1985isometric,winkler1987metric} use the term ``factor'' both for graphs in a factorization and graphs in a pseudofactorization. Here we specifically use the term ``pseudofactors'' to refer to the graphs forming some pseudofactorization.}
Importantly, this task is typically phrased as finding a Cartesian product into which an input graph is isometrically embeddable; however, even $K_2$ has no pseudofactorization into irreducible graphs under this definition (Figure \ref{fig:no-pf}), so it must be generalized carefully to weighted graphs.

\begin{figure}
    \centering
    \includegraphics{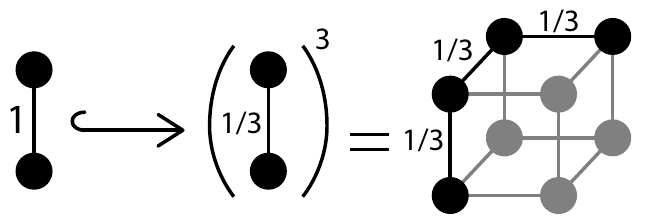}
    \caption{A generalization of pseudofactorization to weighted graphs must be done with care. For example, if pseudofactorization requires only that the input graph be isometrically embeddable into a Cartesian graph product, then even graphs such as $K_2$ have no irreducible pseudofactorization. We require that a graph also be isomorphic to a subgraph of the Cartesian graph product, which precludes this example.}
    \label{fig:no-pf}
\end{figure}

\subsubsection{Isometric embeddings of unweighted graphs}

Despite the connection between pseudofactorization and isometric embedding, studies of isometric embeddings of unweighted graphs preceded Graham and Winkler's original description of pseudofactorization. In 1973, Djokovi\'c \cite{djokovic1973distance} was the first to characterize the unweighted graphs with isometric embeddings into a hypercube, and showed that constructing such an embedding for an unweighted graph can be done in polynomial time. In 1984, Winkler \cite{winkler1984isometric} extended these results to finding isometric embeddings of unweighted graphs into Hamming graphs, designing an $O(|V|^5)$ algorithm to do so. Soon after, Graham and Winkler \cite{graham1985isometric} generalized the earlier results to arbitrary Cartesian graph products, showing that each unweighted graph has a unique isometric embedding into the Cartesian product of its irreducible pseudofactorization. While Graham and Winkler's original paper implied an $O(|E|^2)$ algorithm for finding this unique isometric embedding, later results, such as the $O(|V||E|)$ pseudofactorization algorithm designed by Feder, may be used to construct this isometric embedding as well \cite{feder}.

\subsubsection{Hypercube and Hamming embeddings of weighted graphs}

For some alphabet $\Sigma$ and some integer $k$, a Hamming embedding of a graph $G$ is a map $\eta:V(G)\rightarrow\Sigma^k$ such that for all $u,v\in V(G)$, $d_G(u,v)=d_H(\eta(u),\eta(v))$, where $d_H$ is Hamming distance. A hypercube embedding of $G$ is a Hamming embedding where $\Sigma=\{0,1\}$. The work of Berleant \textit{et al} \cite{joseph_citation} shows that a given graph $G$ is hypercube or Hamming embeddable if and only if its pseudofactors are hypercube or Hamming embeddable.
While it is NP-hard to decide if general weighted graphs are hypercube embeddable, some authors have shown that this can be decided in polynomial time for certain classes of weighted graphs. In particular, if a graph is a line graph or a cycle graph \cite{DezaLaurentL1} or if a graph's distances are all in $\{x,y,x+y\}$ for integers $x,y$, at least one of which is odd \cite{LaurentFewDistances}, it is polynomial time decidable. Shpectorov's results also tell us that for graphs with uniform weights, we can decide this in polynomial time \cite{Shpectorov}. Our results, combined with those of Berleant \textit{et al}, \cite{joseph_citation} show that if there is a polynomial time algorithm to decide if graphs $G_1,G_2,\ldots,G_k$ are hypercube embeddable, then in polynomial time we can also decide hypercube embeddability of any graph $G$ that has the $G_i$ as its pseudofactors.
The number of graphs for which we currently know how to decide hypercube embeddability is very restricted, but if future categories of graphs are found to have polynomial time decidability for this property, this result will apply to extend such a result to the isometric subgraphs of Cartesian product of such graphs with each other and with other graphs in this category.

\subsection{Our results}


In this paper, we start by generalizing the notions of factorization and pseudofactorization to weighted graphs. While defining factorization for weighted graphs is straightforward, pseudofactorization of weighted graphs is more subtle. As an example, if, in analogy to unweighted graphs, a pseudofactorization of graph $G = (V,E,w)$ is defined as a set of graphs for which $G$ is isometrically embeddable into their product, then the graph $K_2$ will have no irreducible pseudofactorization into weighted graphs because it can be isometrically embedded into the Cartesian product of $k$ copies of $K_2$ with edge weight $\frac{1}{k}$ for any positive integer $k$ (Figure \ref{fig:no-pf}). Instead, we require $G$ to be isomorphic to an isometric subgraph of the product of the pseudofactors. This constraint implies both preservation of distances and of edges between $G$ and the pseudofactor product. When all graphs are unweighted, this is equivalent to previous work.

With pseudofactorization defined as such, we are able to prove first that the $O(|E|^2)$ algorithm proposed by Graham and Winkler \cite{graham1985isometric} can be adapted slightly to work on weighted $|E|$-edge graphs. While the algorithm itself is largely unchanged, additional proof is required to show that the output (e.g., the edge weights in the pseudofactors) does not depend on the order in which edges and vertices are traversed, and that the output is a correct pseudofactorization of the input graph into irreducible graphs.

\begin{thm}\label{thm:intro:pseudofactorization-gw}
Given a minimal weighted graph $G=(V,E)$ and the distances between all pairs of vertices, a pseudofactorization into irreducible weighted graphs can be achieved in $O(|E|^2)$ time. If the distances are not pre-computed, the time required to compute all-pairs shortest paths (APSP) must be included, and pseudofactorization may be achieved in $(|V|^2 \log \log |V| + |E|^2)$.
\end{thm}

The APSP running time of $(|V|^2 \log \log |V| + |E||V|)$ is by Pettie~\cite{pettieapsp}, and it is dominated by the $O(|E|^2)$ term for dense enough graphs.

Our proof uses the Djokovi\'c-Winkler relation $\thetarel$ and its transitive closure $\thetahrel$, which are relations on the edges of a graph and are frequently used to pseudofactor unweighted graphs  \cite{winkler1987metric}. For unweighted graphs, each equivalence class $E_i$ of $\thetahrel$ is used to generate one pseudofactor by removing the edges in $E_i$ from the input graph and taking each connected component of the resulting graph to be a vertex of the pseudofactor. Vertices are adjacent in that pseudofactor if there is an edge in $E_i$ connecting the corresponding connected components \cite{graham1985isometric}. To apply this process to weighted graphs, we must prove that all edges connecting any two connected components have the same edge weight (Lemma \ref{claim:pseudofactor1}), and that this edge weight may be used as the edge weight in the corresponding pseudofactor. The proofs that the input graph is isometrically embeddable into the resulting set of pseudofactors, and that each pseudofactor is irreducible, are given in Theorems \ref{claim:pseudofactor5} and \ref{claim:output-irreducible}. The proof of the runtime of this algorithm is given in Section \ref{sec:pseudofactorization-runtime1}.

In addition, we adapt the reasoning of Graham and Winkler on unweighted graphs \cite{graham1985isometric, winkler1987metric} to prove that the irreducible pseudofactorization of a minimal weighted graph is unique. As a result, we call the irreducible pseudofactorization output by this algorithm the \emph{canonical pseudofactorization} and each pseudofactor a \emph{canonical pseudofactor}.

\begin{thm}\label{thm:intro:pseudofactorization-unique}
For a minimal weighted graph, any two pseudofactorizations into irreducible weighted graphs are equivalent in the following sense: there exists a bijection between the two sets of pseudofactors such that corresponding pairs of pseudofactors are isomorphic to each other.
\end{thm}

Theorem \ref{thm:intro:pseudofactorization-unique} is proven in Section \ref{sec:pseudofactorization-unique}.

Finally, we modify a pseudofactorization algorithm on unweighted graphs due to Feder \cite{feder} to speed up the pseudofactorization of minimal weighted graphs. Feder improved upon Graham and Winkler's runtime by finding a spanning tree $T$ for the graph and defining a new  relation $\theta_T$ on the edges of a graph such that two edges in the graph are related by $\theta_T$ if they are related by $\theta$ and at least one of them is in $T$. Feder showed that applying Graham and Winkler's pseudofactorization algorithm using the transitive closure of this relation, $\hat\theta_T$, produces an irreducible pseudofactorization of a weighted graph, no matter the choice of $T$. Because these equivalence classes are quicker to compute than those of $\hat\theta$ and the number of equivalence classes is necessarily limited to $|V|-1$, this results in an improved runtime. While we do not show that the same fact extends to weighted graphs, we are able to use these ideas to get an improved runtime. To do so, we provide Algorithm \ref{alg:faster-alg}, which shows how to find in $O(|E||V|)$ time a spanning tree $T$ of a weighted graph for which $\hat\theta_T$ has the same equivalence classes as $\hat\theta$.  This allows us to improve the time complexity of pseudofactorization to $O(|E||V|)$.

\begin{thm}\label{thm:intro:pseudofactorization-fast}
Given a minimal weighted graph $(V,E)$ and the distances between all pairs of vertices, a pseudofactorization into irreducible weighted graphs is achievable in $O(|E||V|)$ time. If distances are not pre-computed, this is achievable in $O(|V|^2\log\log |V| + |E||V|)$ time.
\end{thm}

Our results on factorization parallel those for pseudofactorization, and in fact, many of our proofs for factorization rely on those for pseudofactorization. Feder showed that an unweighted graph can be factored by replacing the Djokovi\'c-Winkler relation $\thetarel$ with a different relation. We use the replacement relation $\theta\cup\tau$, where $\tau$ relates edges based on a so-called \emph{square property} \cite{feder}. This is similar to the relation of the same name proposed by Feder, but with some added restrictions on the relation between edges on opposite sides of the square in question. This allows us to make the following statement, whose proof follows similar steps to that for pseudofactorization.

\begin{thm}
Given a minimal weighted graph and the distances between all pairs of vertices, a factorization into prime weighted graphs is achievable in $O(|E|^2)$ time. If distances are not pre-computed, this is achievable in $O(|V|^2\log\log |V| + |E|^2)$ time.
\end{thm}

This theorem will be proven in Section \ref{sec:factorization-final-proof} with the runtime analyzed in Section \ref{sec:factorization-runtime}.

As with pseudofactorization, we are also able to show that the prime factorization of a minimal graph $G$ is unique. The proof in this case is much simpler because of the restriction that $G$ be isomorphic to the Cartesian product of its prime factors.

\begin{thm}
For a minimal weighted graph, any two factorizations into prime graphs are equivalent in the following sense: there is a bijection between both sets of factors such that the corresponding pairs of factor graphs are isomorphic.
\end{thm}

The proof of this theorem is given in Section \ref{sec:factorization-unique}.

\subsection{Overview}
Section \ref{sec:preliminaries} describes notation and definitions necessary for the remaining sections. In Section \ref{sec:pseudofactorization}, we show the correctness of the Graham and Winkler algorithm for pseudofactorization of unweighted graphs, as modified for minimal weighted graphs. We also show that each graph has a unique irreducible pseudofactorization up to graph isomorphism. As this algorithm runs in polynomial time, we will be able to use it to make conclusions about properties of graphs later on in the paper. In Section \ref{sec:factorization}, we show that weighted graphs may be factored using a modified version of Feder's algorithm \cite{feder} and prove uniqueness of prime factorization of weighted graphs. Section \ref{sec:runtime-general} introduces our method for analyzing runtimes of the algorithms presented here, and shows that the given algorithms run in $O(|E|^2)$ time plus the time to find all pairs shortest paths (APSP) distances. Then in Section \ref{sec:pseudofactorization-runtime2}, we use a modified version of Feder's algorithm \cite{feder} for pseudofactorization to show that the irreducible pseudofactorization of a minimal weighted graph is computable in as little as $O(|V||E|)$ time plus the time to compute APSP. 
Section \ref{sec:conclusion} concludes with our final thoughts on this work and proposes some remaining open questions.

\section{Preliminaries \label{sec:preliminaries}}

We consider finite, connected, undirected graphs, written $G = (V,E,w)$ with vertex set $V$, edge set $E$, and edge weight function $w : E \to \mathbb{R}_{>0}$. When necessary, we also use $V(G)$, $E(G)$, and $w_G$ to refer to the vertex set, edge set, and weight function of $G$, respectively. For unweighted $G$, we may let $w(e) = 1$ for all $e \in E$. Edges of $G$ are written $uv$ or $vu$ for vertices $u,v \in V$; since all edges are undirected, $uv \in E \iff vu \in E$. The shortest path metric for $G$, written $d_G : V \times V \to \mathbb{R}_{\ge 0}$, maps pairs of vertices to the minimum edge weight sum along a path between them.

\begin{definition}\label{def:minimal-graph}
A graph $G$ is a \emph{minimal graph} if  and only if every edge in $E(G)$ forms a shortest path between its endpoints. That is, $w(uv) = d(u,v)$ for all $uv \in E(G)$.
\end{definition}

Clearly, all unweighted graphs are minimal. We note that for any non-minimal graph, a minimal graph can be generated with the same path metric by simply removing edges not satisfying the minimality condition. Except in Section \ref{sec:factorization} where we assume general weighted graphs, we assume for the remainder of this manuscript that all graphs are minimal. 

A graph embedding $\pi : V(G) \to V(G^*)$ of a graph $G$ into a graph $G^*$ maps vertices of $G$ to those of $G^*$. If $\pi$ satisfies $d_G(u,v) = d_{G^*}(\pi(u),\pi(v))$ all $u,v \in V(G)$, then $\pi$ is an isometric embedding. When such a $\pi$ exists, we say that $G \hookrightarrow G^*$. As a convenience, we let $d_\pi(u,v) = d_{G^*}(\pi(u),\pi(v))$. 

The Cartesian graph product of one or more graphs $G_1, \dots, G_m$ is written $G = G_1 \times \cdots \times G_m$ or $G = \prod_{i=1}^m G_i$. For $G_i = (V_i, E_i, w_i)$, $G$ is defined as $V(G) = V_1 \times \cdots \times V_m$, $E(G)$ is the set of all $(u_1,\dots,u_m)(v_1,\dots,v_m)$ with exactly one $\ell$ such that $u_\ell v_\ell \in E_\ell$ and $u_i=v_i$ for all $i\ne \ell$, and $w_G(uv) = w_\ell(u_\ell v_\ell)$ for $\ell$ chosen as above (Figure \ref{fig:basic-cart-prod}). The Cartesian graph product has an important property regarding its distance metric. For any two vertices $u,v \in V(G)$, $u = (u_1,\dots, u_m)$ and $v = (v_1, \dots, v_m)$, we have that
\begin{equation}\label{eqn:cartesian-distance-property}
d_G(u,v) = \sum_{i=1}^m d_{G_i}(u_i, v_i)
\end{equation}
This property comes about because each edge along a path in $G$ moves in only one of the factors $G_i$, so any path may be decomposed into paths within each of the factors.


We are particularly concerned with cases where $G$ is isometrically embeddable into a Cartesian graph product. The following definition relates edges in $G$ to edges in a Cartesian product into which it has an isometric embedding, and for which edges are preserved between $G$ and the product. 
\begin{definition}\label{def:parent-edge}
Consider graphs $G$ and $G_1,\dots,G_n$ and isometric embedding $\pi : V(G) \to V(\prod_{i=1}^m G_i)$, such that $uv \in E(G) \Rightarrow \pi(u)\pi(v) \in E(\prod_{i=1}^m G_i)$ for all $u,v \in V(G)$. For any edge $uv\in E(G)$, $\pi(u) = (u_1, \dots, u_m)$ and $\pi(v) = (v_1, \dots, v_m)$, there exists exactly one $\ell$ such that $u_\ell\neq v_\ell$, and we must have $u_\ell v_\ell \in E(G_\ell)$. We call $u_\ell v_\ell$ the \textbf{parent edge} of $uv$ under $\pi$. If $G$ equals the product of the $G_i$, then we may implicitly assume $\pi$ to be the identity isomorphism.
\end{definition}

Note that from our definition of Cartesian product, edge $uv \in E(G)$ with parent edge $u_\ell v_\ell \in E(G_\ell)$ must have $w_G(uv) = w_{G_\ell} (u_\ell v_\ell)$.

The term factorization is used when $G$ is isomorphic to a Cartesian graph product (e.g., Figure \ref{fig:factor}b).
\begin{definition}\label{def:factorization}
Whenever $G$ is isomorphic to a Cartesian graph product of $G_i$, $1\le i\le m$, we say that the set $\{G_1, \dots, G_m\}$ forms a \textbf{factorization} of $G$ and refer to each $G_i$ as a \textbf{factor}.
\end{definition}




If all factorizations of $G$ include $G$ as a factor, we say that $G$ is \emph{prime}. A prime factorization is one with only prime factors. For convenience, we assume that a factorization does not include $K_1$, except in the case where $G = K_1$, since a factor of $K_1$ does not affect the final product.

Pseudofactorization generalizes factorization to situations where $G$ is not isomorphic to the graph product. Instead, we require that $G$ only be isomorphic to an isometric subgraph of the graph product.

\begin{definition}\label{def:pseudofactorization}
Consider graphs $G$ and $G^* = \prod_{i=1}^m G^*_i$. If an embedding $\pi : V(G) \to V(G^*)$, $\pi = (\pi_1, \dots, \pi_m)$ exists satisfying the following criteria:
\begin{enumerate}
\item $d_G(u,u') = d_{G^*}(\pi(u), \pi(u'))$,
\item $uv \in E(G)$ implies $\pi(u)\pi(v) \in E(G^*)$ and $w_G(uv) = w_{G^*}(\pi(u)\pi(v))$,
\item every vertex in $G^*_i$ is in the image of $\pi_i$, $1 \le i \le m$, and
\item every edge in $G^*_i$ is the parent of an edge in $G$
\end{enumerate}
then we say the set $\{ G^*_1, \dots, G^*_m \}$ is a \textbf{pseudofactorization} of $G$ and refer to each $G^*_i$ as a \textbf{pseudofactor}.
\end{definition}

If all pseudofactorizations of $G$ include $G$ as a pseudofactor, we say that $G$ is \emph{irreducible}. An irreducible pseudofactorization is one with only irreducible pseudofactors. As with factorization, we assume that a pseudofactorization does not include $K_1$, except in the case where $G = K_1$.

Clearly, any pseudofactorization is also a factorization; however, the converse is not true (see Figure \ref{fig:factor}c). Informally, the definition of pseudofactorization requires both that $G$ be isometrically embeddable into ${G^*}$ \emph{and} that edges be preserved within this embedding. This second condition is a natural one for manipulating graph structures, but may be less applicable to other situations (e.g., finite metric spaces). The final two conditions ensure that there are no unnecessary vertices and edges in the pseudofactors (or any graph would be a pseudofactor of the graph in question, as $G$ is an isometric subgraph of $G\times H$ for any graph $H$).

\section{Pseudofactorization of weighted graphs}\label{sec:pseudofactorization}
In this section, we will discuss a method for pseudofactoring weighted graphs in polynomial time. To begin, we discuss the current state of the field in terms of pseudofactoring unweighted graphs, and we then show that one of the techniques used for this process can also be used to pseudofactor weighted graphs.

Graham and Winkler \cite{graham1985isometric} showed that all unweighted graphs have a unique pseudofactorization. They additionally gave an $O(|E|^2)$ time algorithm to find this pseudofactorization.  To do so, they defined the $\theta$ relation on the edges of a graph as follows. 

For a graph $G=(V,E)$, two edges in the graph, $uv,ab\in E$ are related by $\theta$ if and only if: 
\begin{align}
    [d_G(u,a)-d_G(u,b)]-[d_G(v,a)-d_G(v,b)] \neq 0. \label{eqn:theta-diff}
\end{align}
We note that this relation is symmetric and reflexive. We also let the equivalence relation $\hat\theta$ be the transitive closure of $\theta$.
We call the left side of equation \ref{eqn:theta-diff} the theta-difference for edges $uv$ and $ab$.

Algorithm \ref{alg:general-alg} is a generalized version of the algorithm presented by Graham and Winkler. Its inputs are a graph and an equivalence relation on the edges of the graph, and it outputs a set of graphs. Graham and Winkler showed that when the input is $(G,\hat \theta)$ for an unweighted graph $G$, the output is an irreducible pseudofactorization of $G$. Figure \ref{fig:alg-rep} shows an example of an application of this algorithm to a weighted graph when the input relation is $\hat\theta$.

\begin{figure}
    \centering
    \includegraphics[width=\linewidth]{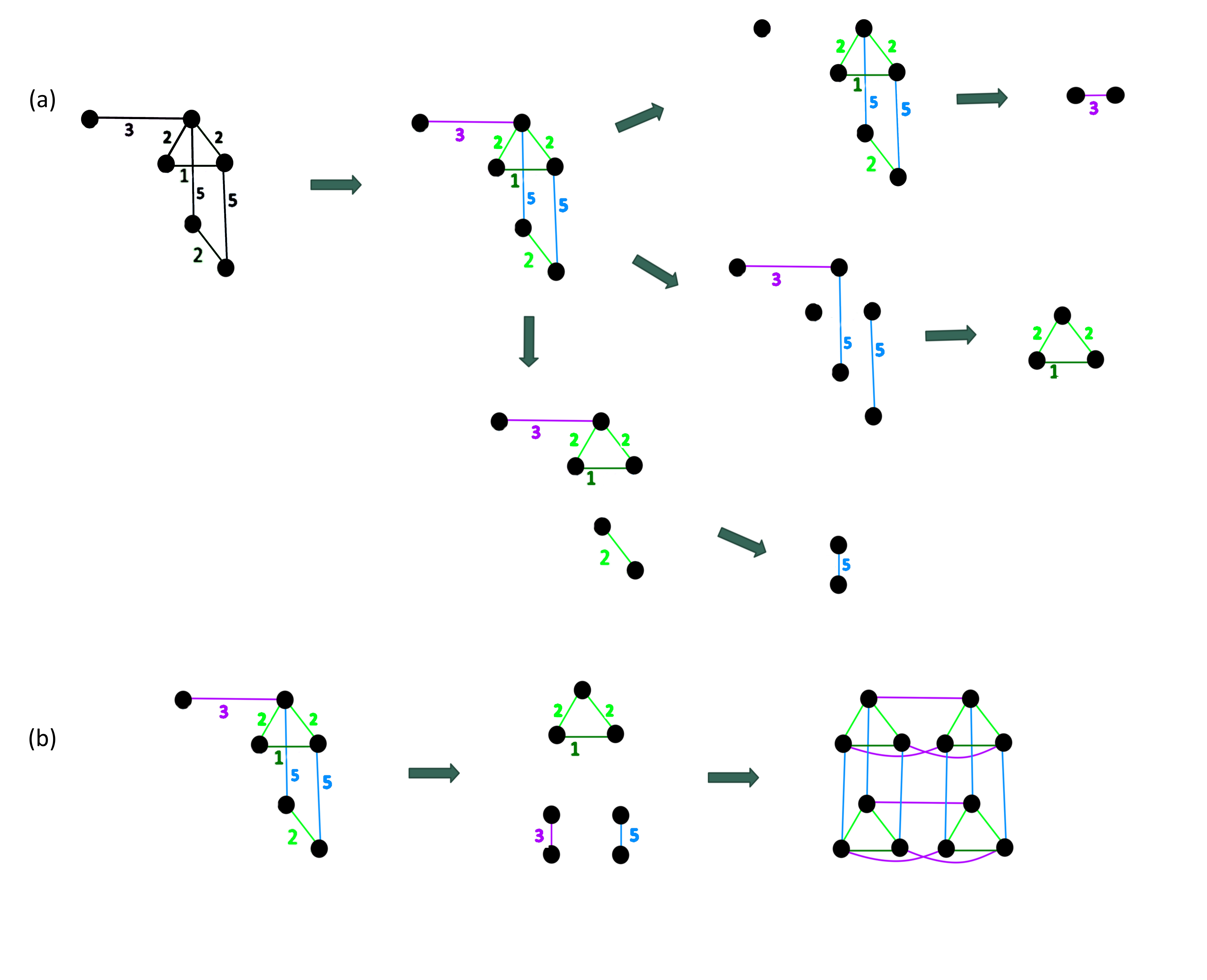}
    \caption[ Illustration of Algorithm \ref{alg:general-alg} on a weighted graph of six nodes.]
    {
    Illustration of Algorithm \ref{alg:general-alg} on a weighted graph of six nodes. \\
    (a) The algorithm first finds the equivalence classes of $\hat\theta$, shown here with distinct colors. The algorithm then considers the remaining graph when each equivalence class is removed and uses that to construct the final output graphs. In this case, there are three $\hat\theta$ equivalence classes and thus three pseudofactors.\\
    (b) This subfigure shows in the first panel the same graph as in (a), in the second panel it shows its irreducible pseudofactorization, and in the third panel it shows Cartesian product of those pseudofactors, of which the original graph is an isometric subgraph.
    \label{fig:alg-rep}
    }
    
\end{figure}

Informally, the algorithm finds the equivalence classes of the graph edges under the given relation. For each equivalence class $E_k$, it then looks at the subgraph of $G$ with all edges in $E_k$ removed. The graph may now be disconnected, and the algorithm uses this disconnected graph to construct one of the output graphs. In this paper, we will expand on this to show that the same algorithm works for weighted graphs. If $G$ contains multiple edges of different weights between a pair of connected components in any of these new graphs, the algorithm defined here will reject, but for the purposes of this paper, we will only use equivalence relations for which this never happens.

\begin{algorithm}
\caption{Algorithm for breaking up a graph over a relation}
\label{alg:general-alg}
\textbf{Input}: A weighted graph $G$ and an equivalence relation $R$ on the edges of $G$. 

\textbf{Output}: A set \textbf{G}=$\{G_1^*,G_2^*,...,G_m^*\}$

\begin{algorithmic}
    \State  Set \textbf{G} $\leftarrow \emptyset$
    \State Find the set of equivalence classes of $R$, $\{E_1,E_2,...,E_m\}$
    \State Set \textbf{E} $\leftarrow\{E_1,E_2,...,E_m\}$
    \For{$E_k\in$\textbf{E}}
        \State Let $w_{G_k'}$ be $w_G$ restricted to the edges in $E(G) \setminus E_k$
        \State Set $G_k' \leftarrow (V(G),E(G)\setminus E_k,w_{G_k'})$
        \State Set \textbf{C} $\leftarrow \{C_1, C_2, \dots, C_\ell\}$ (the set of connected components of $G_k'$)
        \State Create a new graph $G_k^*$
        \State Set $V(G_k^*) \leftarrow \{a\suchthat C_a \in  \text{\textbf{C}}\}$
        \State Set $E(G_k^*) \leftarrow \{ab\suchthat \text{ there is an edge in } E_k \text{ between } C_a \text{ and } C_b \}$
            
        \For{$ab \in E(G_k^*)$} 
            \If{All edges between $C_a$ and $C_b$ have the same weight, $w_{ab}$,} 
            \State Set $w_{G_k^*}(ab)\leftarrow w_{ab}$
            \Else
            \State Reject
            \EndIf
        \EndFor

        \State Set \textbf{G}$\rightarrow$ \textbf{G}$\cup \{G_k^*\}$
    \EndFor
    \State Return \textbf{G}.
\end{algorithmic}
\end{algorithm}

Graham and Winkler showed that an unweighted pseudofactorization is unique. Additionally, $\hat\theta$ is not the only relation that can be used with this algorithm to produce this pseudofactorization. In particular,
Feder \cite{feder} expanded on this work by defining a new relation, $\hat\theta_T$ on the edges of a graph $G$ given a spanning tree $T$ of $G$. In particular, he defined $\theta_T$ such that $uv \ \theta_T \ ab$ if and only if $\edgediff{G}{u}{v}{a}{b}\neq 0$ and at least one of $uv$ and $ab$ is in $T$. Letting $\hat\theta_T$ be the transitive closure of $\theta_T$, he showed that Algorithm \ref{alg:general-alg} on input $(G,\hat\theta_T)$ for an unweighted graph $G$ also produces the irreducible pseudofactorization of $G$, noting that for this relation Algorithm \ref{alg:general-alg} runs in $O(|V||E|)$ time due to the shorter time needed to find the equivalence classes. In a later section, we will discuss how Feder's algorithm may be modified to improve the runtime of pseudofactorization, but in this section we will focus on the application of Graham and Winkler's algorithm to weighted graphs.




In this section, the overall goal is to prove that when the input to Algorithm \ref{alg:general-alg} is a minimal weighted  graph $G$ and the relation $\hat\theta$, the output is an irreducible pseudofactorization of $G$. 

\subsection{Testing irreducibility}
We first show that if all edges in a graph are in the same equivalence class of $\hat\theta$, then the graph is irreducible. In the following section, we will show that Algorithm \ref{alg:general-alg} with this relation as an input produces a pseudofactorization of the graph. Because this algorithm produces more than one graph (neither of which is $K_1$) if there is more than one equivalence class on the edges, we can use the results in this section and the next to conclude that checking the number of equivalence classes of $\hat\theta$ is a definitive check of irreducibility.

\begin{lem} \label{claim:same-pseudofactor}
For $uv,xy\in E(G)$, if $uv\thetahrel xy$, then for any pseudofactorization $\{G_1,\ldots,G_m\}$ of $G$ and isometric embedding $\pi:V(G)\rightarrow V(\Pi_{i=1}^mG_i)$, $uv$ and $xy$ must have parent edges under $\pi$ in the same pseudofactor. 
\end{lem}
\begin{proof}
Throughout this proof, we let $d:=d_G$ and $d_i:=d_{G_i}$. Say there exists a pseudofactorization $\{G_1,G_2,...,G_m\}$ of $G$  with such an embedding $\pi$ such that $\pi(a)=(\pi_1(a),\pi_2(a),...,\pi_m(a))$ for $a\in V(G)$. For simplicity, we let $\pi_i(a)=a_i$. 

We first prove the lemma for $uv \thetarel xy$. Assume for contradiction that $xy$ and $uv$ have parent edges under $\pi$ in different pseudofactors. Since $xy$ and $uv$ are edges, by Definition \ref{def:pseudofactorization}, $\pi(x)\pi(y)$ and $\pi(u)\pi(v)$ are also edges, and by the definition of the Cartesian product, we get that there is exactly one $l$ such that $u_l\neq v_l$ and one $j$ such that $x_j\neq y_j$. Since the two edges have parent edges in different pseudofactor graphs, we also get that $l\neq j$ (see Definition \ref{def:parent-edge}).
    
Now we consider $[d(x,u)-d(x,v)]-[d(y,u)-d(y,v)]$. Since 
$\pi$ is an isometric embedding with $\pi_i(x)=x_i$, we can rewrite this using the distance metric for $\Pi_iG_i^*$, which means it can be written as:
\begin{align*}
    \sum_{i=1}^m [d_{i}(x_i,u_i)-d_{i}(x_i,v_i)]-[d_{i}(y_i,u_i)-d_{i}(y_i,v_i)].
\end{align*}
Term $i$ in this sum is 0 if $u_i=v_i$ or if $x_i=y_i$. However, since $l\neq j$, this means at least one of these equalities is true for every term in the sum, so $xy$ and $uv$ are not related by $\theta$. From this, we get that if $xy\thetarel uv$, then $l=j$ and $xy$ and $uv$ must have parent edges under $\pi$ in the same pseudofactor graph.

To prove the lemma when $uv \thetahrel xy$, observe that $uv \thetahrel xy$ implies that there is a sequence of edges, $e_1 = uv, e_2, \dots, e_l = xy$ for which $e_k \thetarel e_{k+1}$ for all $1 \le k < l$. By the above reasoning, parent edges under $\pi$ for adjacent pairs of edges belong to the same pseudofactor, so the same is true for $uv$ and $xy$.
\end{proof}

We know that Algorithm \ref{alg:general-alg} outputs one graph for each equivalence class of $\hat\theta$. Thus, from the preceding lemma, if Algorithm \ref{alg:general-alg} outputs a single graph then the input graph must be irreducible. In the following section, we will show that the output of this algorithm is necessarily a pseudofactorization of the input graph, which together with this proof implies that a graph is irreducible if and only if there is one equivalence class of $\hat\theta$ on its edges.

\subsection{An algorithm for pseudofactorization}
In this section, we will show that Algorithm \ref{alg:general-alg} with $\hat\theta$ as the input relation can be used to pseudofactor a minimal weighted graph. Many of the lemmas used in this section have parallels to those that we will use to prove factorization. 
First, we show in Lemma \ref{claim:pseudofactor1} that this algorithm is well-defined for the inputs we are considering (a minimal graph and the relation $\hat \theta$). Throughout this section, we assume the input graph is $G$ and notation is used as it is in Algorithm \ref{alg:general-alg}.

\begin{lem}\label{claim:pseudofactor1}
If $C_a,C_b$ are connected components in $G_k'$ and there exists $x\in C_a,y\in C_b$ such that $xy$ is an edge with weight $w_{ab}$, then for each $u\in C_a$ there exists \textit{at most} one $v\in C_b$ such that $uv$ forms an edge, and if it exists the edge has weight $w_{ab}$.
\end{lem}

\begin{proof}
Throughout this proof, we take $d(\cdot,\cdot)$ to be the distance function on $G$. First, we show that if $uv$ is an edge between $C_a,C_b$, then there cannot exist a distinct $v'\in C_b$ such that $uv'$ is an edge. We will show this by contradiction. Assume that such a $v'$ exists. Since $v$ and $v'$ are in the same connected component of $G_k'$, there is a path $Q=(v=q_0,q_1\ldots,q_n=v')$ (represented in Figure \ref{fig:well-defined}(a)) consisting entirely of edges not in $E_k$ (and thus not related to $uv$ by $\theta$), and we consider the sum
\begin{align*}
    \sum_{i=1}^n[d(u,q_{i-1})-d(u,q_i)]-[d(v,q_{i-1})-d(v,q_i)]=0.
\end{align*}
By telescoping, this implies that:
\begin{align*}
    [d(u,v)-d(u,v')]-[d(v,v)-d(v,v')] &= d(u,v)+d(v,v')-d(u,v') = 0.
\end{align*}
Since for $v\neq v'$,  $d(v,v')>0$, this gives us $d(u,v)<d(u,v')$. However, a symmetric analysis says $d(u,v')<d(u,v)$, so $u$ cannot have edges to two distinct $v,v' \in C_b$. This shows the first part of the lemma.

Now, we show that if $uv$ is an edge from $C_a$ to $C_b$, then it has weight $w_{ab}$, which will rely on the assumption that the graph in question is minimal. First, we define two paths. The first is $Q_a=(u=q_0^a,q_1^a,\ldots,q_n^a=x)$, which is a path of edges entirely in $C_a$. We will also have a path $Q_b=(v=q_0^b,q_1^b,\ldots,q_t^b=y)$, which will consist entirely of edges in $C_b$ (represented in Figure \ref{fig:well-defined}(b)). We note that no pair of edges on either of these paths can be related to $uv$ or to $xy$ by $\theta$, since the edges are not in $E_k$.
Using this fact, we get the following four sums.
\begin{align*}
    \sum_{l=1}^t[d(u,q_{l-1}^b)-d(u,q_{l}^b)]-[d(v,q_{l-1}^b)-d(v,q_{l}^b)] &= 0 \\
    &= [d(u,v)-d(u,y)]-[d(v,v)-d(v,y)] \\
    &= d(u,v)-d(u,y)+d(v,y) \\
    &= w_G(uv)+d(v,y)-d(u,y).\\
    \sum_{l=1}^n[d(u,q_{l}^a)-d(u,q_{l-1}^a)]-[d(v,q_{l}^a)-d(v,q_{l-1}^a)] &= 0 \\
    &= [d(u,x)-d(u,u)]-[d(v,x)-d(v,u)] \\
    &= d(u,x)+d(v,u)-d(v,x)\\
    &= w_G(uv)+d(u,x)-d(v,x). \\
    \sum_{l=1}^n[d(x,q_{l-1}^a)-d(x,q_{l}^a)]-[d(y,q_{l-1}^a),d(y,q_{l}^a)] &= 0 \\
    &= [d(x,u)-d(x,x)]-[d(y,u)-d(y,x)] \\
    &= d(u,x)-d(y,u)+d(y,x)\\
    &= w_G(uv)+d(y,x)-d(y,u). \\
    \sum_{l=1}^t[d(x,q_{l}^b)-d(x,q_{l-1}^b)]-[d(y,q_{l}^b)-d(y,q_{l-1}^b)] &= 0 \\
    &= [d(x,y)-d(x,v)]-[d(y,y)-d(y,v)] \\
    &= d(x,y)-d(x,v)+d(v,y) \\
    &= w_G(uv)+d(v,y)-d(x,v).
\end{align*}

From these equations, we get:
\begin{align*}
    w_G(uv) &= d(u,y)-d(v,y) \\
    w_G(uv) &= d(v,x)-d(u,x) \\
    w_G(xy) &= d(u,y)-d(u,x) \\
    w_G(xy) &= d(v,x)-d(v,y).
\end{align*}
Subtracting the first and last of these equations gives us $w_G(uv)-w(xy)=d(u,y)-d(v,x)$ and subtracting the second and third equations gives $w_G(uv)-w_G(xy)=-[d(u,y)-d(v,x)]$. This gives us that $w_G(uv)-w_G(xy)=-[w_G(uv)-w_G(xy)]$, so the difference between these two weights is 0 and thus the weights are equal. This implies the lemma.
\end{proof}

\begin{figure}
    \centering
    \includegraphics[width=.7\linewidth]{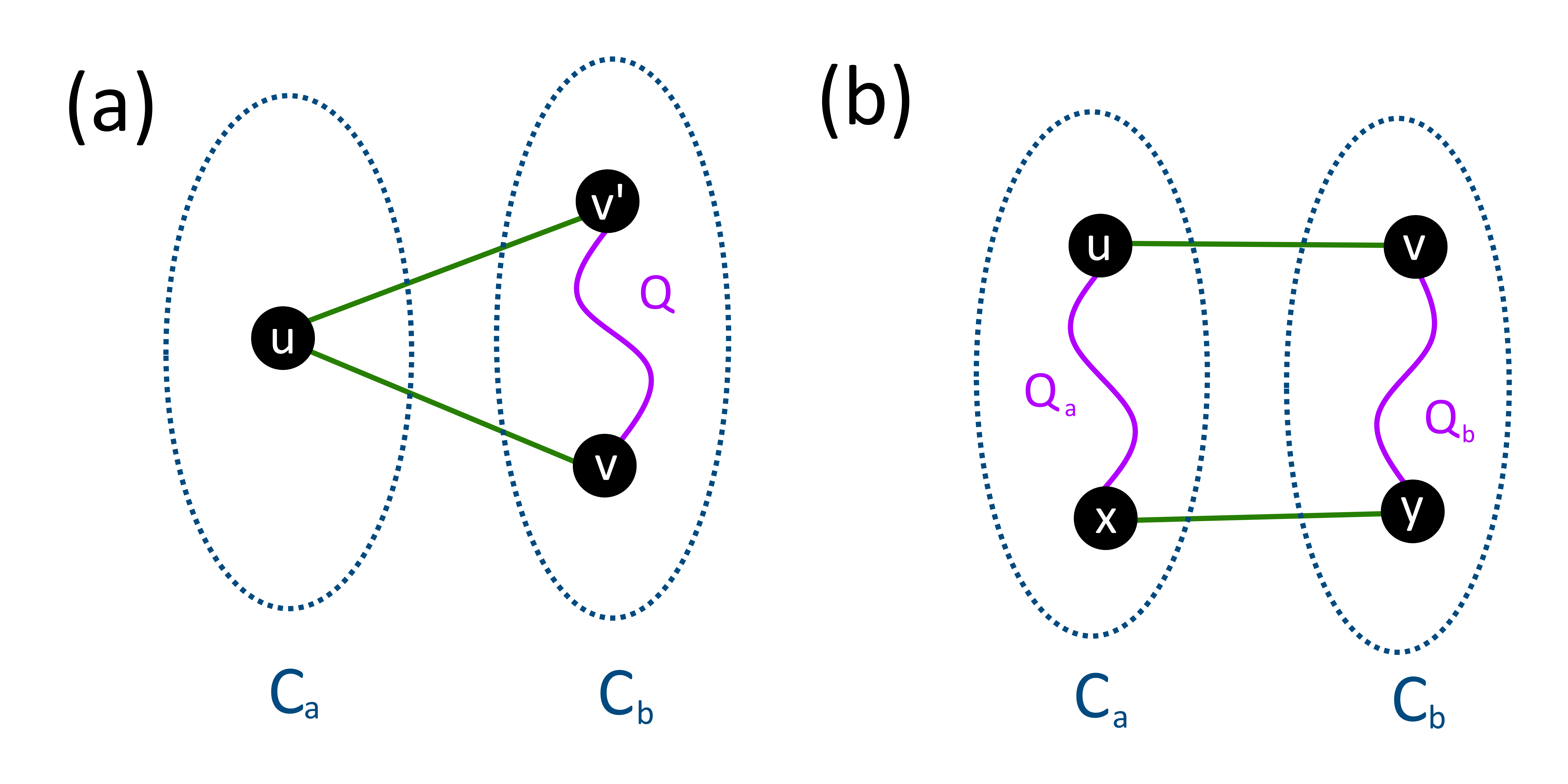}
    \caption[An illustration of the paths in Lemma \ref{claim:pseudofactor1}]{Paths described in the proof of Lemma \ref{claim:pseudofactor1}. (a) A graph in which $u\in C_a$ has an edge to distinct $v,v'\in C_b$, which the proof of Lemma \ref{claim:pseudofactor1} shows is impossible. (b) Definitions of $Q_a$ and $Q_b$ from the proof when considering $u\in C_a,v\in C_b$.}
    \label{fig:well-defined}
\end{figure}



With the next lemma, we introduce two general facts about the relationship between paths and equivalence classes of $\hat\theta$ that will be used in later claims.

\begin{lem}\label{claim:pseudofactor2}
The following hold:

\begin{enumerate}
    \item If $uv$ forms an edge and is in equivalence class $E_k$, then for any path $Q=(u=q_0,q_1,\ldots,q_t=v)$ between the two nodes, there is at least one edge from $E_k$.
    \item Let $P=(u=p_0,p_1,\ldots,p_n=v)$ be a shortest path from $u$ to $v$. If $P$ contains an edge in the equivalence class $E_k$, then for any path $Q=(u=q_0,q_1,\ldots,q_t=v)$ there is at least one edge from $E_k$.
\end{enumerate}
\end{lem}

\begin{proof}
First, we note that the second bullet implies the first in the case of unweighted and minimal graphs, since in those cases any edge $uv$ is a shortest path between $u$ and $v$. However, this is not the case for general weighted graphs and we will use this lemma when we discuss factorization of weighted graphs as well. Thus, we prove this lemma in two parts.

\begin{enumerate}
    \item First, we consider the following sum.
    \begin{align*}
        \sum_{i=1}^t[d(u,q_{i-1})-d(u,q_i)]-[d(v,q_{i-1})-d(v,q_i)] &= [d(u,u)-d(u,v)]-[d(v,u)-d(v,v)] \\ 
        &=-2d(u,v) \\
        &\neq 0
    \end{align*}
    The last inequality comes from the fact that $u\neq v$ and the assumption that the graph has only positive weight edges. However, we note that this sum is only non-zero if at least one term in the sum is non-zero. If term $i$ is non-zero, then $uv$ and $q_{i-1}q_i$ are related by $\theta$ and $q_{i-1}q_i$ is thus in $E_k$. Thus, we prove the first part of the lemma.
    
    \item Fix an edge $p_{\ell-1}p_\ell$ in $P$ and consider the following sum.
    \begin{align*}
        \sum_{i=1}^t[d(p_{\ell-1},q_{i-1})-d(p_{\ell-1},q_i)]-[d(p_\ell,q_{i-1})-d(p_\ell,q_i)] \\ = [d(p_{\ell-1},u)-d(p_{\ell-1},v)]-[d(p_\ell,u)-d(p_\ell,v)]
    \end{align*}
    We know that because $P$ is a shortest path, $d(u,v)=d(u,p_{\ell-1})+d(v,p_{\ell-1})$ and $d(u,v)=d(u,p_\ell)+d(v,p_\ell)$. Substitution yields:
    \begin{align*}
        d(u,v) - 2d(p_{\ell-1},v) - d(u,v) + 2d(p_\ell,v) = 2[d(p_\ell,v)-d(p_{\ell-1},v)].
    \end{align*}
    This value is only 0 if $v$ is equidistant from $p_\ell$ and $p_{\ell-1}$, but since they form an edge on the shortest path between $u$ and $v$, this is not possible and thus at least one edge on $Q$ is related to $p_{\ell-1}p_\ell$ by $\theta$ and thus at least one edge in $Q$ is in its equivalence class.
\end{enumerate}
\end{proof}

Now, we move on to the goal of showing that $G$ is isomorphic to an isometric subgraph of $\Pi_iG_i^*$. To do so, we define a mapping $\pi$ from $G$ to $\Pi_iG_i^*$ with the goal to show that $\pi$ is an injection and that shortest paths in $G$ correspond to shortest paths under $\pi$. This will help us show our target property about isometric subgraphs.

First we will define $\pi$ more formally. We will let $\pi:V(G)\rightarrow V(\Pi_iG_i^*)$ and will define $\pi$ as: $\pi(u)=(\pi_1(u),\pi_2(u),\ldots,\pi_m(u))$. We then must define each $\pi_i$, so we let $\pi_i(u)$ be the node in $G_i^*$ that corresponds to the connected component of $G_i'$ that $u$ is a member of. We notice that Algorithm \ref{alg:general-alg} only includes a node in $G_i^*$ if there is a node in the corresponding connected component of $G_i'$ and only includes an edge in $G_i^*$ if there is an edge in $G$ between nodes in the corresponding connected components of $G_i'$. Because of this, we see that criteria 3 and 4 of Definition \ref{def:pseudofactorization} are met by this mapping.
We now show that $\pi$ is an injection. 

\begin{claim}\label{claim:psuedofactor3}
As defined in this section, $\pi$ is an injection.
\end{claim}
\begin{proof}
We show that $\pi(u)\neq \pi(v)$ for all $u,v\in V(G)$, $u\ne v$. To do so, let $P$ be a shortest path between $u$ and $v$. We know that if one of the edges is in $E_k$, then $u$ and $v$ are in different connected components of $G_k'$ because Lemma \ref{claim:pseudofactor2} says that all paths between the two nodes have an edge in $E_k$. Since $u\neq v$ there is at least one edge in $P$, and thus there is at least one index $k$ for which $\pi_k(u)\neq \pi_k(v)$. 
\end{proof}

In many of our proofs, manipulation of the sum of theta-differences over a path (which we will call a theta-sum) is essential. We will now show an important property of that sum that we will use in our final proof. Informally, it says that for the theta-sum along a path, the contribution from the edges in each equivalence class does not depend on the path taken. 

We define the following notation for paths in a graph. For path $P=(u=p_0,p_1,\ldots,p_n=v)$ in $G = (V,E,w)$, let $P_k$ be the sequence of edges in $P$ that are also in $E_k$, where $E_k$ is one of the equivalence classes of $E$ under $\thetahrel$. Define $T_k^P$ as $T_k^P:=\sum_{p_ip_{i+1}\in P_k}[d(u,p_i)-d(u,p_{i+1})]-[d(v,p_i)-d(v,p_{i+1})]$. 

\begin{lem}\label{claim:theta-sum}
Let $P=(u=p_0,p_1,\ldots,p_n=v)$ and $Q=(u=q_0,q_1,\ldots,q_t=v)$ be two paths in $G$ from $u$ to $v$. Then $T_k^P=T_k^Q$.
\end{lem}

\begin{proof}
We consider the following equations.
\begin{align*}
    T_k^P &= \sum_{p_ip_{i+1}\in P_k}[d(u,p_i)-d(v,p_i)]-[d(u,p_{i+1})-d(v,p_{i+1})] \\
    &=\sum_{p_ip_{i+1}\in P_k}\sum_{q_jq_{j+1}\in Q}[d(q_j,p_i)-d(q_{j+1},p_i)]-[d(q_{j},p_{i+1})-d(q_{j+1},p_{i+1})] \\
    &= \sum_{p_ip_{i+1}\in P_k}\sum_{q_jq_{j+1}\in Q_k}[d(q_j,p_i)-d(q_{j+1},p_i)]-[d(q_{j},p_{i+1})-d(q_{j+1},p_{i+1})] \\
    &= \sum_{q_jq_{j+1}\in Q_k}\sum_{p_ip_{i+1}\in P_k}[d(q_j,p_i)-d(q_{j+1},p_i)]-[d(q_{j},p_{i+1})-d(q_{j+1},p_{i+1})] \\
    &=\sum_{q_jq_{j+1}\in Q_k}\sum_{p_ip_{i+1}\in P_k}[d(q_j,p_i)-d(q_{j},p_{i+1})]-[d(q_{j+1},p_{i})-d(q_{j+1},p_{i+1})] \\
    &= \sum_{q_jq_{j+1}\in Q_k}\sum_{p_ip_{i+1}\in P}[d(q_j,p_i)-d(q_{j},p_{i+1})]-[d(q_{j+1},p_{i})-d(q_{j+1},p_{i+1})] \\
    &= \sum_{q_jq_{j+1}\in Q_k}[d(q_j,u)-d(q_{j},v)]-[d(q_{j+1},u)-d(q_{j+1},v)] \\
    &= T_k^Q
\end{align*}
The second equality comes from telescoping the inner sum, the third comes from the fact that only edges related by $\theta$ can contribute to the sum, so the only edges that might contribute are those in $E_k$. The fourth equality comes from switching the order of the sums, the fifth comes from switching the order of the terms in the summand, and the sixth again from the fact that only edges in $E_k$ contribute. The seventh equality is by telescoping, and the last equality is by definition of $T_k^Q$.
\end{proof}

From here, we are able to prove our overall goal using Theorem \ref{claim:pseudofactor5}

\begin{thm}\label{claim:pseudofactor5}
For any two nodes $u,v\in V(G)$, there is a shortest path between them in $G$ that under $\pi$ is a shortest path in $\Pi_iG_i^*$. Because $G$ is minimal, this implies that $\pi$ maps $G$ to a subgraph of $\Pi_iG_i^*$ in which all edges in $E(G)$ are preserved under the map and the distance metric is preserved in the subgraph, meaning that criteria 1 and 2 of Definition \ref{def:pseudofactorization} are met by this mapping. Thus, for an input $(G,\hat\theta)$, Algorithm \ref{alg:general-alg} produces a pseudofactorization of $G$.
\end{thm}

\begin{proof}
First, we show that if $P$ is a shortest path in $G$ then the sequence of nodes in $P$ under $\pi$ forms a path in $\Pi_iG_i^*$. We do so by showing that if there is an edge $ab\in E(G)$, then there is an edge $\pi(a)\pi(b)\in E(\Pi_iG_i^*)$ and the edge has the same weight. 

Assume $ab\in E_j$. First, let $C_{\pi_j(a)}$ and $C_{\pi_j(b)}$ be the connected components of $G_j'$ corresponding to nodes $\pi_j(a)$ and $\pi_j(b)$ in $G_j^*$. We know that $E(G)$ has an edge between a node in $C_{\pi_j(a)}$ and a node in $C_{\pi_j(b)}$ of weight $w_G(ab)$ (by Lemma \ref{claim:pseudofactor1}). By the same lemma, this means that all edges between nodes in $C_{\pi_j(a)}$ and nodes in $C_{\pi_j(b)}$ have the same weight and, by the definition of $G_j^*$, there is an edge between $\pi_j(a)$ and $\pi_j(b)$ of that weight. This holds for all edges, which means that any path in $G$ still exists in the image of $\pi$ in $\Pi_iG_i^*$, so for any pair $u,v\in V(G)$, $d(u,v)\geq d^*(u,v)$, where $d^*(\cdot,\cdot)$ is the distance metric for $\Pi_iG_i^*$.

Now, assume $d(u,v) > d^*(u,v)$ for at least one pair of vertices $u,v \in V(G)$. Let $u,v\in V(G)$ be a pair of nodes with the smallest value of $d^*(\pi(u),\pi(v))$ such that $d^*(\pi(u),\pi(v))<d(u,v)$. 

First, we show that any node on a shortest path from $\pi(u)$ to $\pi(v)$ in the product graph cannot be in the image of $\pi$. Consider a node $\pi(u')$ that is in the image of $\pi$ and on such a shortest path. By the definition of shortest path, we get that $d^*(\pi(u),\pi(u'))+d^*(\pi(u'),\pi(v))=d^*(\pi(u),\pi(v))$. Additionally, we have $d^*(\pi(u),\pi(u'))=d(u,u')$ because $d^*(\pi(u),\pi(u'))<d^*(\pi(u),\pi(v))$ and we get $d^*(\pi(u'),\pi(v))=d(u',v)$ for parallel reasons. Then we get a contradiction: 
\begin{align*}
    d(u,v) &\leq d(u,u')+d(u',v) \\
    &= d^*(\pi(u),\pi(u')) +d^*(\pi(u'),\pi(v)) \\
    &= d^*(\pi(u),\pi(v)) \\
    &< d(u,v).
\end{align*}
The first inequality is the triangle inequality and the last line is by the assumption. Since we have reached a contradiction, we know such a $u'$ cannot exist and thus no node in the image of $\pi$ can be on the shortest path between $\pi(u)$ and $\pi(v)$ in $G^*$.

We have shown that no node in the image of $\pi$ can be on the shortest path between $\pi(u)$ and $\pi(v)$, but to get a contradiction, we will now show that such a node \textit{must exist}. 


Let $uu'$ be the first edge on a shortest path from $u$ to $v$ and let the edge be in $E_k$. If $P$ is the path in question, we show that $P_k$ is a shortest path in $G_k^*$. As discussed in Section \ref{sec:preliminaries}, this implies that $\pi(u)\pi(u')$ is an edge on a shortest path from $\pi(u)$ to $\pi(v)$ in $\Pi_iG_i^*$. Consider a path $Q^*=(\pi_k(u)=q_0,q_1,\ldots,q_n=\pi_k(v))$ between $\pi_k(u)$ and $\pi_k(v)$. Since the nodes in $G_k^*$ are defined to be the labels of connected components in $G_k'$, each of which consists of one or more nodes in $G$, for each $j$ there exists $v_j$ such that $\pi_k(v_j)=q_j$ (i.e. there is a node in $C_{q_j}$ in $G_k'$). We also have that because $q_jq_{j+1}$ forms an edge, there is a pair of nodes in $G$ that has an edge between $C_{q_j}$ and $C_{q_{j+1}}$. Let $(x_j^j,x_{j+1}^j)$ be this pair. 
We have as a superscript the index of the edge we're considering and as a subscript the index of the connected component in $G_k'$.

We will construct the path $Q$ from $u$ to $v$ in $G$ shown in Figure \ref{fig:construct-path}(a). Let $Q_j$ be a path from $x_j^{j}$ to $x_{j+1}^j$ that does not include any edges in $E_k$ and let $Q_0$ be a path from $u$ to $x_0^1$ and $Q_n$ be a path from $x_n^{n}$ to $v$ without any edges from $E_k$. (Since each of the pairs is in the same connected component of $G_k'$, these paths must exist.) Construct the path $Q$ from $u$ to $v$ in $G$ such that $Q:=Q_0+Q_1+Q_2+\cdots + Q_n=(u=f_0,f_1,\ldots,f_m=v)$.  A visualization of this construction is given in Figure \ref{fig:construct-path}. This is a path from $u$ to $v$ and we consider the sum $T_k^Q=\sum_{f_jf_{j+1}\in Q_k}[d(v,f_j)-d(v,f_{j+1})]-[d(u,f_j)-d(u,f_{j+1})]$.

\begin{figure}
    \centering
    \includegraphics[width=.8\linewidth]{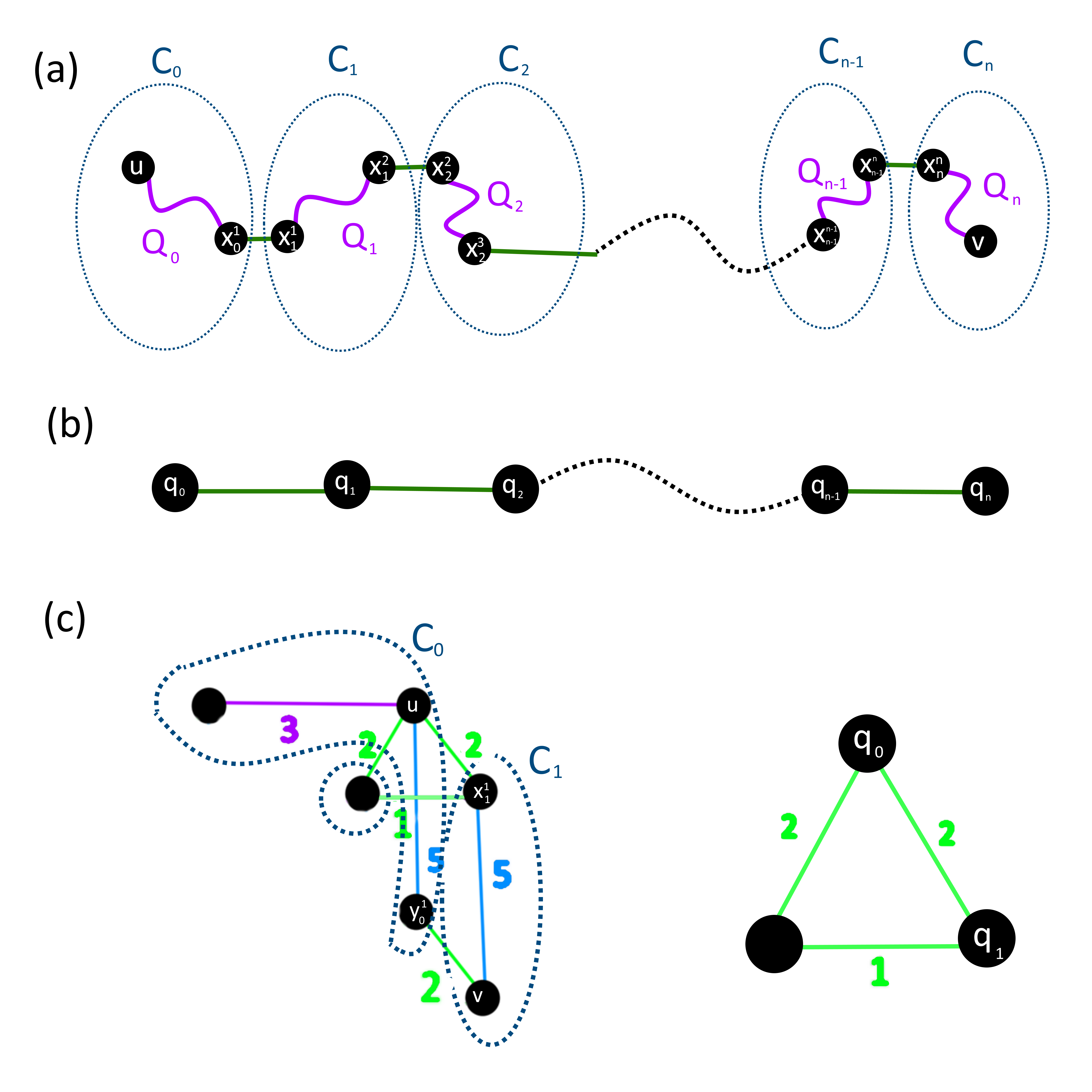}
    \caption[A visualization of the path constructed in Theorem \ref{claim:pseudofactor5}]{A visualization of the path constructed in the proof of Theorem \ref{claim:pseudofactor5}.
    
    (a) Each $C_i$ is a connected component in $G_k'$, where $E_k$ is the equivalence class we consider in the theorem. Starting at $u$, we construct a path $Q_0$ through edges in $C_0$ to $x_0^1$, which is any node in $C_0$ with an edge to  $C_1$. We similarly construct $Q_1$ from $x_1^1$ to $x_1^2$ and so on. The dotted line here represents an arbitrary length path that is not shown.  
    
    (b) Part of $G_k^*$ for the graph in (a). We have $\pi_k(u)=q_0$ and $\pi_k(v)=q_n$, and the path $Q^*=(q_0,q_1,\ldots,q_n)$ is an arbitrary path from $\pi_k(u)$ to $\pi_k(v)$ through $G_k^*$. The theorem shows that because of how $G_k^*$ is constructed, there must exist a path in $G$ like the one shown in (a), where the only edges in $E_k$ are the $x_{j-1}^jx_{j}^j$ and the weight of each $x_{j-1}^jx_{j}^j$ in $G$ is the same weight as $q_{j-1}q_j$ in $G_k^*$.
    
    (c) An example of this process applied to the graph from Figure \ref{fig:alg-rep}. In particular, we consider the pseudofactor shown on the right. When the graph on the left is isometrically embedded into the product of its pseudofactors (which is shown in Figure \ref{fig:alg-rep}) with isometric embedding $\pi$, $\pi_k(u)=q_0$ and $\pi_k(v)=q_1$. We let $Q^*=(q_0,q_1)$. When constructing $Q$, we can select $Q_0$ to be $(u)$ (so $x_0^1=u$) and $Q_1$ to be $(x_1^1,v)$. Alternatively, we could select $Q_0$ to be $(u,y_0^1)$ and $Q_1$ to be $(v)$ (so $v=y_1^1$). Thus, we see that while the path we construct in this theorem exists, it is not necessarily unique.
    }
    \label{fig:construct-path}
\end{figure}







We know from Lemma \ref{claim:theta-sum} that $T_k^P=T_k^Q$. We use this to get
\begin{align*}
    T_k^P &= T_k^Q \\
    &= \sum_{q_iq_{i+1}\in Q_k}[d(v,q_i)-d(v,q_{i+1})]-[d(u,q_i)-d(u,q_{i+1})] \\
    &\leq \sum_{q_iq_{i+1}\in Q_k}2w(q_iq_{i+1}).
\end{align*}
Since $P$ is a shortest path, every edge must contribute 2 times its weight to the overall theta-sum (or else we would not be able to get to the total) so we get $T_k^P=\sum_{p_ip_{i+1}\in P_k}2w(p_ip_{i+1})\leq \sum_{q_iq_{i+1}\in Q_k}2w(q_iq_{i+1})$. Thus, the path $P_k$ from $u_k$ to $v_k$ is a shortest path in $G_k^*$. This means $u_ku_k'$ is a first edge on a shortest path from $\pi_k(u)$ to $\pi_k(v)$ in $G_k^*$, which means $\pi(u)\pi(u')$ is a first edge on a shortest path from $\pi(u)$ to $\pi(v)$ in $\Pi_iG_i^*$ - a contradiction. 

Thus, we can conclude $d(u,v)=d^*(\pi(u),\pi(v))$ and we already showed all edges are preserved, so $G$ is an isometric subgraph of $\Pi_iG_i^*$.
\end{proof}

Using the previous theorem, we conclude that Algorithm \ref{alg:general-alg} produces a set of graphs for which the input graph is isometrically embeddable into their Cartesian product. Additionally, we can show that this is actually an irreducible pseudofactorization by showing that all of the produced pseudofactors are irreducible.

\begin{lem}\label{claim:output-irreducible}
If a minimal graph $G$ is the input graph to Algorithm \ref{alg:general-alg} and $\hat \theta$ is the input relation, the output graphs are all irreducible.
\end{lem}
\begin{proof}
Using Lemma \ref{claim:same-pseudofactor}, we need only show that the edges in each factor are all related by $\hat \theta$ (when evaluated on the factor itself). Consider two edges $ab,a'b'\in E(G)$ where $ab\thetahrel a'b'$. We know that there must exist a sequence of edges $ab=a^1b^1,a^2b^2,\ldots,a^tb^t=a'b'$ such that $a^ib^i\thetarel a^{i+1}b^{i+1}$ for $i\in \{1,2,\ldots,t-1\}$. Because they're all related by $\hat\theta$, we know that all edges in this sequence must have parents in the same pseudofactor, which we will call pseudofactor $j$. Let $u_j:=\pi_j(u)$ for all $u\in V(G)$. We only need to show that for each $i$, $a^i_jb^i_j\thetarel a_j^{i+1}b_{j}^{i+1}$ in order to show that $a_jb_j\thetahrel a_j'b_j'$. 

For simplicity, we will just consider two edges $xy,uv\in E(G)$ such that $xy\thetarel uv$ and show that if $j$ is the pseudofactor where the two edges have parent edges under $\pi$, then $x_jy_j\thetarel u_jv_j$. We get the following:

\begin{align*}
    [d^*(\pi(x),\pi(u))-d^*(\pi(x),\pi(v))]-[d^*(\pi(y),\pi(u))-d^*(\pi(y),\pi(v))]\\
    =\sum_{i=1}^m[d_i(\pi_i(x),\pi_i(u))-d_i(\pi_i(x),\pi_i(v))]-[d_i(\pi_i(y),\pi_i(u))-d_i(\pi_i(y),\pi_i(v))] \\
    = [d_j(\pi_j(x),\pi_j(u))-d_j(\pi_j(x),\pi_j(v))]-[d_j(\pi_j(y),\pi_j(u))-d_j(\pi_j(y),\pi_j(v))] \\=
    [d_j(x_j,u_j)-d_j(x_j,v_j)]-[d_j(y_j,u_j)-d_j(y_j,v_j)],
\end{align*}
where the second to last equality is due to the fact that $\pi(x)\pi(y)$ is an edge with a parent in $j$ and the last equality is due to how we defined $x,y,u,v$.
Thus, we get that any pair of edges in $G_j^*$ is related by $\hat\theta$.
\end{proof}

Thus, from Theorem \ref{claim:pseudofactor5}  and Lemma \ref{claim:output-irreducible}, we conclude that Algorithm \ref{alg:general-alg} with the input $(G,\hat\theta)$ for a minimal weighted $G$ produces an irreducible pseudofactorization of $G$.






\subsection{Uniqueness of pseudofactorization}\label{sec:pseudofactorization-unique}

Here, we prove that the irreducible pseudofactorization of any minimal graph is unique in the following sense: for any two irreducible pseudofactorizations of a minimal graph $G$, the two sets of pseudofactors are equal up to graph isomorphism. We call the irreducible pseudofactorization generated by Algorithm \ref{alg:general-alg} with input $(G, \thetahrel)$ the \emph{canonical pseudofactorization} or \emph{canonical pseudofactors} of $G$.

We note that uniqueness is guaranteed in part because conditions 3 and 4 of Definition \ref{def:pseudofactorization} require that no unnecessary vertices or edges are included in the pseudofactors. Of course, if these conditions are removed, an arbitrary number of vertices and edges may be added without affecting isometric embeddability into the product and the uniqueness property no longer holds. (In fact, removing condition 3 makes it so that no graph is irreducible.)

\begin{thm}\label{thm:pseudofactorization-unique}
Let $G$ be a minimal weighted graph and let $G_1, \dots, G_p$ and $H_1, \dots, H_r$ be two irreducible pseudofactorizations of $G$. Then $p=r$ and the pseudofactors may be reordered so that $G_i$ is isomorphic to $H_i$ for all $i, 1 \le i \le p$.
\end{thm}

\begin{proof}
This proof follows the reasoning of Graham and Winkler \cite{graham1985isometric}, with modifications for weighted graphs.

Number the $\thetahrel$ equivalence classes of $E(G)$ as $E_1, \dots, E_{m}$. Let $\pi : V(G) \to V\left(\prod_{i=1}^pG_i\right)$ and $\rho : V(G) \to V\left(\prod_{i = 1}^r H_i\right)$ be isometric embeddings of $G$ into $\prod_{i =1}^pG_i$ and $\prod_{i=1}^r H_i$, respectively, with $\pi = (\pi_1, \dots, \pi_{p})$ and $\rho = (\rho_1, \dots, \rho_{r})$. Note that by the definition of pseudofactorization, if $u, v \in V(G)$ are adjacent then $\pi(u), \pi(v)$ are adjacent, and there is exactly one $i$ such that $\pi_i(u) \ne \pi_i(v)$; edge $\pi_i(u)\pi_i(v)$ is the parent edge under $\pi$ of $uv \in E(G)$. 
The same reasoning applies to $\rho(u)$ and $\rho(v)$.

First, we show that  $uv \thetahrel u'v'$ if and only if they have parent edges belonging to the same pseudofactor. That is, $uv \thetahrel u'v'$ if and only if $\pi_i(u) \ne \pi_i(v) \iff \pi_i(u') \ne \pi_i(v')$ for all $i, 1 \le i \leq p$.
For the forward case, when $uv \thetahrel u'v'$ then Lemma \ref{claim:same-pseudofactor} applies and we are done. For the reverse case, let $uv$ and $u'v'$ have parent edges in $G_j$ and note that because $G_{j}$ is irreducible, it has exactly one equivalence class of $\thetahrel$ on its edges, since Algorithm \ref{alg:general-alg} outputs a pseudofactorization with one pseudofactor per equivalence class of $\thetahrel$. Further, we have that
\begin{align}
& \edgediff{G}{u}{v}{u'}{v'} \\
&= \edgediff{\pi}{u}{v}{u'}{v'} \\
&= \sum_{i=1}^p \edgediff{\pi_i}{u}{v}{u'}{v'} \\
&= \edgediff{\pi_j}{u}{v}{u'}{v'}
\end{align}
where the first equality holds because $\pi$ is an isometric embedding, the second equality holds because of the path decomposition property of the Cartesian graph product, and the third equality holds because each summand is nonzero only if $\pi_i(u') \ne \pi_i(v')$. Thus, by the definitions of $\theta$ and of pseudofactorization, if all edges in $G_{j}$ are related by $\thetahrel$ then any edges in $G$ to which they are a parent are also related by $\thetahrel$.

This shows there is a bijection from the $\thetahrel$ equivalence classes to the set of $G_i$, as well as to the set of $H_i$. So $p=r=m$.

Renumber both pseudofactorizations so that $G_i$ and $H_i$ contain only parent edges under $\pi$ and $\rho$, respectively, of the edges in $E_i$, $1 \le i \le m$. Fix $j$, $1 \le j \le m$, and $u \in V(G)$. Take any $u'$ for which there is a path $P$ to $u$ not using any edge in $E_j$. Clearly, $\pi_j(u) = \pi_j(u')$ because no edges in this path have parent edges in $G_j$ and so $\pi_j$ is constant along this path. Alternatively, take any $u'$ for which every path to $u$ has at least one edge in $E_j$ and consider a shortest path $P$ from $u$ to $u'$. Let $c_i$ be the sum of the edge weights of the edges along $P$ in $E_i$. Observe:
\begin{equation}
d_G(u,u') = \sum_{i=1}^n d_{\pi_i}(u,u') = \sum_{i=1}^n c_i
\end{equation}
Each $d_{\pi_i}(u,u') \le c_i$ because the edges along $p$ in $E_i$ trace out a path in $G_i$. Thus, for the sums to be equal $d_{\pi_i}(u,u') = c_i$. So in this case $d_{\pi_j}(u,u') = c_j > 0$ and $\pi_j(u) \ne \pi_j(u')$. Thus, $u$ and $u'$ are connected by a path without edges in $E_j$ if and only if $\pi_j(u) = \pi_j(u')$. Now let $V_u$ be the set of all $u' \in V(G)$ for which there is a path from $u$ to $u'$ without an edge in $E_j$. Then by this reasoning $V_u = \{u' \in V(G) : \pi_j(u) = \pi_j(u')\} = \{u' \in V(G) : \rho_j(u) = \rho_j(u')\}$. Thus, there is a single $\rho_j(u) \in V(H_j)$ for each $\pi_j(u) \in V(G_j)$. Let $f : V(G_j) \to V(H_j)$ map each vertex of $G_j$ to the corresponding vertex in $H_j$ given by $V_u$.

Each edge $\pi_j(u)\pi_j(v)$ in $G_j$ is the parent of an edge in $G$ (see condition 4 of Definition \ref{def:pseudofactorization}), which we assume without loss of generality to be $uv \in E_j$. As observed above, $uv \in E_j$ must have a parent edge under $\rho$, $\rho_j(u)\rho_j(v) \in E(H_j)$. From the definition of pseudofactorization, the edge weights of $\pi_j(u)\pi_j(v)$, $uv$, and $\rho_j(u)\rho_j(v)$ must be equal. Thus, for every edge $\pi_j(u)\pi_j(v) \in E(G_j)$ there is an edge $f(\pi_j(u))f(\pi_j(v)) = \rho_j(u)\rho_j(u') \in E(H_j)$, and, by symmetry, the converse is true. So $G_j$ is isomorphic to $H_j$.
\end{proof}

\section{Factorization of weighted graphs}\label{sec:factorization}

Feder \cite{feder} showed that Algorithm \ref{alg:general-alg} could also be used to find the factorization of a graph. To do so, he defined a new relation, $\tau_1$. For edges $uv,uv'\in E(G)$, $uv\ \tau_1 \ uv'$ if and only if there does not exist a 4-cycle containing edges $uv$ and $uv'$. He further defined $(\theta\cup\tau_1)^*$ to be the transitive closure of $(\theta\cup\tau_1)$ and showed that for an unweighted graph $G$, Algorithm \ref{alg:general-alg} on the input $(G,(\theta\cup\tau_1)^*)$ produces the prime factorization of $G$. 
We note that in this section, we are discussing a \textit{general} weighted graph, not necessarily a minimal one.

\subsection{A modified equivalence relation}

For the purpose of graph factorization, we first define a property which we call the \textit{square property} and which will be useful for factorization. (This is a modification of the square property proposed by Feder.)

\begin{definition}\label{def:square-prop} Edges $uv,uv'\in E(G)$ satisfy the \textbf{square property} if there exists a vertex $x$ such that $uvxv'$ forms a square (four-cycle) with $w_G(uv)=w_G(xv')$, $w_G(uv')=w_G(xv)$, $uv\thetarel xv'$, and $uv'\thetarel xv$. In other words, the two edges must make up two of the adjacent edges of at least one four-cycle in which the opposite edges have the same weight and are related by $\theta$. This is depicted visually in Figure \ref{fig:square-property}.
\end{definition}

\begin{figure}
    \centering
    \includegraphics[width=.5\linewidth]{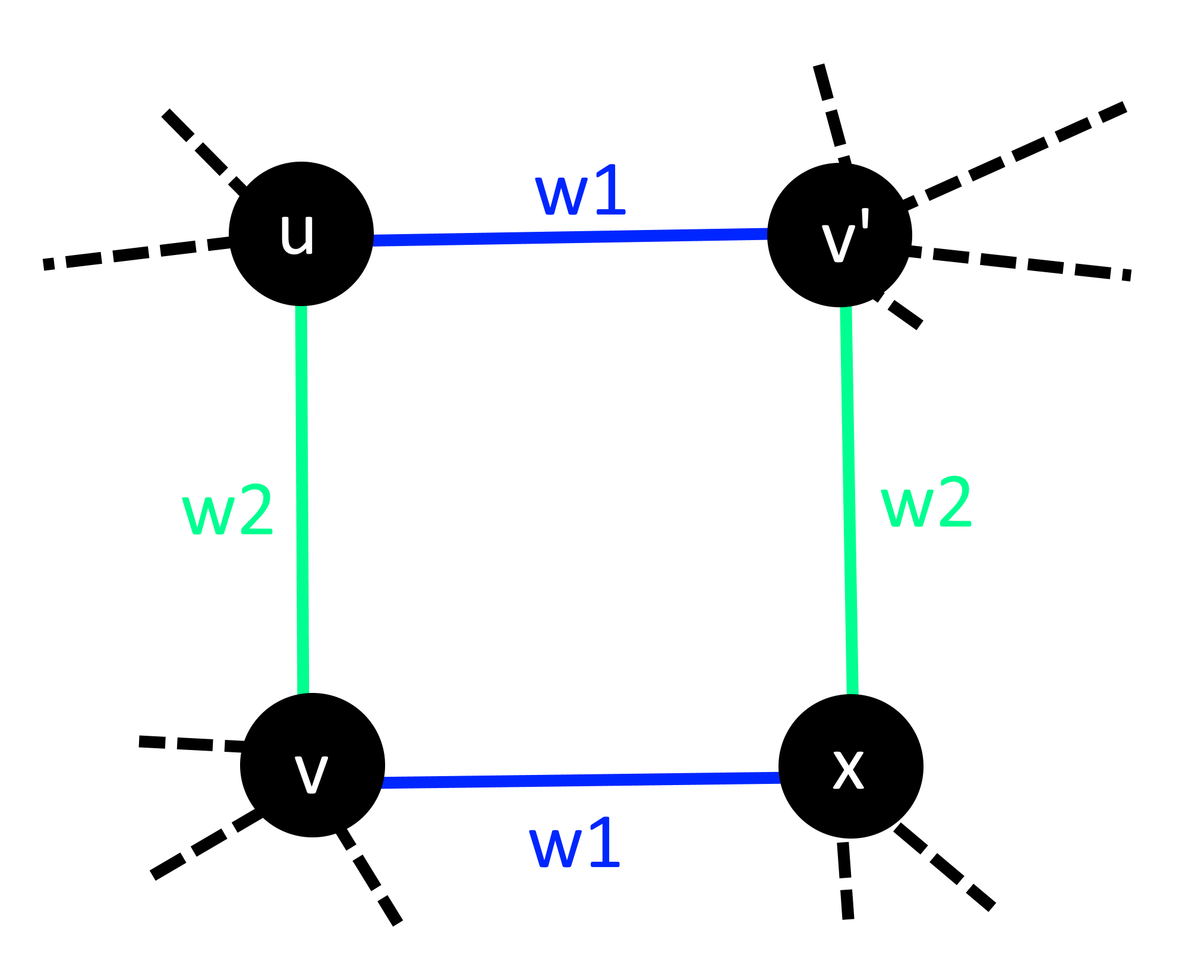}
    \caption[A visual representation of the square property]{If $uv'\thetarel vx$ and $uv\thetarel v'x$, then existence of the subgraph depicted here implies that $uv$ and $uv'$ satisfy the square property.}
    \label{fig:square-property}
\end{figure}

We now define two relations $\theta$ and $\tau$ as follows:

\begin{enumerate}
    \item As before, two edges, $xy,uv\in E(G)$ are related by $\theta$ if and only if $[d_G(x,u)-d_G(x,v)]-[d_G(y,u)-d_G(y,v)]\neq 0$. 
    
    \item Two edges $uv,uv'\in E(G)$ are related by $\tau$ if and only if they do \textit{not} satisfy the square property. If two edges do not share a common endpoint, they are not related by $\tau$.
\end{enumerate}

We note that the $\theta$ relation used here is the same as that used by Graham and Winkler \cite{graham1985isometric} and the $\tau$ relation is inspired by that used by Feder for factorization of unweighted graphs \cite{feder}.
We also let $\transthetarel$ be the transitive closure of $\theta$ and $\factorrel$ be the transitive closure of $(\theta\cup \tau)$. Using these definitions, we prove that we can factor graphs using Algorithm \ref{alg:general-alg}. 

\subsection{Testing primality}

Analogously to showing irreducibility with respect to pseudofactorization, in this section, we show in Lemma \ref{claim:same-factor} that a graph is prime if it has one equivalence class under $(\theta\cup\tau)^*$ on its edges. 

\begin{lem} \label{claim:same-factor}
For $uv,xy\in E(G)$, if $uv\thetarel xy$ or $uv\taurel xy$, then for any factorization of $G$, $uv$ and $xy$ must have parent edges in the same factor under $\alpha$, where $\alpha$ is an isomorphism from $G$ to the product of the factors. From this, we get that if $uv\factorrel xy$ then $uv$ and $xy$ have parent edges under $\alpha$ in the same factor.
\end{lem}
\begin{proof}
We will prove the lemma first for the $\theta$ relation and then for the $\tau$ relation. Throughout the proof of this lemma, we let $d:=d_G$ and $d_i:=d_{G_i}$. Say there exists a factorization $\{G_1,G_2,...,G_m\}$ of $G$  with isomorphism $\alpha$ such that $\alpha(a)=(\alpha_1(a),\alpha_2(a),...,\alpha_m(a))$ for $a\in V(G)$. For simplicity, we let $\alpha_i(a)=a_i$. 

\begin{enumerate}
    \item Consider $uv\thetarel xy$. Then by Lemma \ref{claim:same-pseudofactor}, we get that in any embedding into an isometric subgraph of the product graph, the images of these two edges have parents in the same factor. Since $\alpha$ is an isomorphism and thus is an isometric embedding into a subgraph of the product, this claim applies and we get that the parent edges under $\alpha$ of $uv$ and $xy$ must be in the same factor.
    
    
    
    \item Consider $uv\taurel uv'$ and assume for contradiction that the above factorization of $G$ is one in which $uv$ and $uv'$ have parent edges under $\alpha$ in different factors.
    We get that there is exactly one $l$ such that $u_l\neq v_l'$ and exactly one $j$ such that $u_j\neq v_j$.  Since they belong to different factor graphs, we have $l\neq j$. Without loss of generality, assume $j<l$.
    
    Now, consider the node $(u_1,\ldots,u_{j-1},v_j,u_{j+1},\ldots,u_{l-1},v_l',u_{l+1},\ldots,u_m)$
    in $\Pi_iG_i$ (ie the node that matches $\alpha(u)$ on all indices except $j$ and $l$, where it matches $v_j$ and $v_l'$ respectively).
    Since the vertex set of $\Pi_iG_i$ is the Cartesian product of the $V(G_i)$, this node must be in $\Pi_iG_i$ and since $\alpha$ is a bijection, there must exist $x\in V(G)$ such that $\alpha(x)$ equals this node. We also know that $x$ must have an edge to $v$ since $x_i=v_i$ for all $i\neq l$ and  
    $v_lv_l'=u_lv_l'\in E(G_l)$. It also has an edge to $v'$ since $x_i=v_i'$ for all $i\neq j$ and $v_jv_j'=v_ju_j\in E(G_j)$.
    
    Thus, $uvxv'$ is a square. Additionally, we know that  weights on opposite sides of the square are equal because they have the same parent edge under $\alpha$. We can also show that opposite edges are related by $\theta$. We will only show this for $uv\thetarel xv'$ and appeal to symmetry to show $uv'\thetarel xv$. As before, we can rewrite $[d(u,x)-d(u,v')]-[d(v,x)-d(v,v')]$ as the sum:
    \begin{align*}
        \sum_{i=1}^m[d_i(u_i,x_i)-d_i(u_i,v_i')]-[d_i(v_i,x_i)-d_i(v_i,v_i')]
    \end{align*}
    Since $u$ and $v$ only differ on coordinate $j$, this becomes: 
    \begin{align*}
        [d_{j}(u_j,x_j=v_j)-d_{j}(u_j,v_j'=u_j)]-[d_{j}(v_j,x_j=v_j
    )-d_{j}(v_j,v_j'=u_j)] &= 2d_{j}(u_j,v_j) \\
     &\neq  0.
    \end{align*}
     The last inequality comes from the fact that $v_j\neq u_j$ and we assume that all edge weights are positive. Thus, $uv\thetarel xv'$ and by a symmetric argument $uv'\thetarel xv$. This means $x$ is such that $uv$ and $uv'$ satisfy the square property, which means they cannot be related by $\tau$.
\end{enumerate}

Thus, if two edges are related by $\theta$ or by $\tau$, then they have parent edges in the same factor for any factorization of $G$. This property is preserved under the transitive closure, so if two edges are related by $\factorrel$ then they must correspond to edges in the same factor.
\end{proof}

In the following section, we will show that Algorithm \ref{alg:general-alg} with $(\theta\cup\tau)^*$ as the input relation gives a prime factorization of the input graph. Since the algorithm outputs one graph for each equivalence class of the given relation, combined with Lemma \ref{claim:same-factor}, this tells us that a graph is prime if and only if all of its edges are in the same equivalence class of $(\theta\cup\tau)^*$.

\subsection{An algorithm for factorization\label{sec:factorization-final-proof}}


In this section, we show that Algorithm \ref{alg:general-alg} with inputs $G$ and $(\theta\cup \tau)^*$ produces a prime factorization of $G$. We first show that the algorithm in question is well-defined for general weighted graphs and with the input relation $(\theta\cup \tau)^*$. We use Lemma \ref{claim:factorization1} to show this, and we note that it actually proves a  stronger statement that we will continue to use later. Additionally, we note that throughout this section, we will refer to $G_k'$ and $G_k^*$ as they are defined in the algorithm. 
\begin{lem}\label{claim:factorization1}
If $C_a,C_b$ are connected components in $G_k'$ and there exists $x\in C_a,y\in C_b$ such that $xy$ is an edge with weight $w_{ab}$, then for each $u\in C_a$ there exists exactly one $v\in C_b$ such that $uv$ forms an edge and that edge has weight $w_{ab}$.
\end{lem}
\begin{proof}
Throughout this proof, we take $d(\cdot,\cdot)$ to be the distance function on $G$. First, we have that if $uv$ is an edge between $C_a,C_b$, then there cannot exist a distinct $v'\in C_b$ such that $uv'$ is an edge. This proof is identical to the first part of the proof of Lemma \ref{claim:pseudofactor1} so we do not repeat it here.





Now we have that a node in $C_a$ cannot have edges to more than one node in $C_b$, and we expand on that to show that each $u\in C_a$ has an edge to \textit{at least one} node in $C_b$ and that that edge has weight $w_{ab}$. Let $P(u)$ be a path from $u$ to $x$ that does not include any edges in $E_k$ and has the smallest number of edges of all such paths. (Such a path must exist because $u$ and $x$ are in the same connected component of $G_k'$.) 

We proceed by induction on the number of edges in the path $P(u)$. In the base case, there are zero edges in $P$. In this case, we must have $u=x$. By assumption, $x$ has an edge of weight $w_{ab}$ to $y\in C_b$, proving the inductive hypothesis. 

Now, we assume that for all nodes in $C_a$ with such a path $P$ of $n$ edges, the lemma holds. We let $u$ be a node such that $P(u)$ has $n+1$ edges and let $u'$ be the second node on this path. By definition, we know that $P(u')$ has $n$ edges. By the inductive assumption, there exists $v'\in C_b$ such that $u'v'\in E(G)$ and with $w(u'v')=w_{ab}$. We know that $uu'$ and $u'v'$ are not related by $\tau$ (or else $uu'$ would be in $E_k$), so they must fulfill the square property. Let  $v$ be the node such that $uu'v'v$ is a square with opposite edges of equal weight and relatd by $\theta$. Because $uu'\not\in E_k$, we get $vv'\notin E_k$ and since $v'\in C_b$, we get that $v\in C_b$. This means $uv$ is an edge between $C_a$ and $C_b$ with weight $w_{ab}$, proving the lemma.
\end{proof}

We now write Lemma \ref{claim:factorization2}, which is identical to Lemma \ref{claim:pseudofactor2}, but now refers to the equivalence classes of our new relation. We note that the proof of this claim is identical to that of the original claim, so we do not repeat it here. 

\begin{lem}\label{claim:factorization2}
We have the following two facts about the equivalence classes of $(\theta\cup\tau)^*$.

\begin{enumerate}
    \item If $uv$ forms an edge and is in equivalence class $E_k$, then for any path $Q=(u=q_0,q_1,\ldots,q_t=v)$ between the two nodes, there is at least one edge from $E_k$.
    \item Let $P=(u=p_0,p_1,\ldots,p_n=v)$ be a shortest path from $u$ to $v$. If $P$ contains an edge in the equivalence class $E_k$, then for any path $Q=(u=q_0,q_1,\ldots,q_t=v)$ there is at least one edge from $E_k$.
\end{enumerate}
\end{lem}

We now want to show that $G$ is ismorphic to $\Pi_iG_i^*$, so we define an isomorphism $\alpha:V(G)\rightarrow V(\Pi_iG_i^*)$. For ease of notation, we let $\alpha(u)=(\alpha_1(u),\alpha_2(u),\ldots,\alpha_m(u))=(u_1,u_2,\ldots,u_m)$. We define the function such that $\alpha_i(u)$ is the name of the connected component of $G_i'$ that $u$ is a member of. This is the same way we defined $\pi$ in the previous section, but because we are using a new equivalence relation, we show that we now have an isomorphism. We now want to show that $\alpha$ is a bijection and that $uv$ is an edge if and only if $\alpha(u)\alpha(v)$ is an edge (and that they have the same weight if so), which will show that $\alpha$ is an isomorphism. We use Lemma \ref{claim:factorization3} to show the first part of this fact.

\begin{lem}\label{claim:factorization3}
As defined in the preceding paragraph, $\alpha$ is a bijection.
\end{lem}
\begin{proof}
First, we show that $\alpha$ is an injection (i.e. $\alpha(u)\neq \alpha(v)$ for all $u,v\in V(G)$, $u\ne v$). This is identical to the proof that $\pi$ is an injection in Lemma \ref{claim:psuedofactor3}, but we reiterate it here using the terminology in this section. To show $\alpha$ is an injection, let $P$ be a shortest path between $u$ and $v$. We know that if one of the edges is in $E_k$, then $u$ and $v$ are in different connected components of $G_k'$ because Lemma \ref{claim:factorization2} says that all paths between the two nodes have an edge in $E_k$. We know that if $u\neq v$ then there is at least one edge in $P$ and thus there is at least one index $k$ on which $\alpha_k(u)\neq \alpha_k(v)$.

Now, we show that $\alpha$ is a surjection by showing that all nodes in $\Pi_iG_i$ are in the image of $\alpha$. Assume for contradiction that there is at least one node in the product not in the image of $\alpha$. Let $(u_1,u_2,\ldots,u_m)$ be one such node with an edge to a node $(u_1,\ldots,x_k,\ldots,u_m)$  in the image of $\alpha$ whose pre-image is $x$. Since the two nodes have an edge of some weight $w_{xu}$ between them, we know that there is an edge between $C_{u_k}$ and $C_{x_k}$ of weight $w_{xu}$ and Lemma \ref{claim:factorization1} tells us that every node in $C_{x_k}$ has an edge of that weight to exactly one node in $C_{u_k}$. Let $u'$ be the node in $C_{u_k}$ that $x$ has an edge to in $G$. Because there is an edge between $u'$ and $x$, we know that they appear in the same connected component for all $G_i'$ with $i\neq k$. This tells us $\alpha_i(x)=\alpha_i(u')$ for all $i\neq k$ and thus $\alpha(u')=(u_1,u_2,\ldots,u_m)$, which means $(u_1,u_2,\ldots,u_m)$ is in the image of $\alpha$.
\end{proof}

Finally, we use the next theorem to show that $\alpha$ is an isomorphism. 
\begin{thm}\label{claim:factorization4}
For $u,v\in V(G)$, $u\ne v$, $uv\in E(G)\iff \alpha(u)\alpha(v)\in E(\Pi_iG_i^*)$. If they do form edges, they have the same weight. Thus, Algorithm \ref{alg:general-alg} on an input $(G,(\theta\cup\tau)^*)$ outputs a factorization of $G$.
\end{thm}

\begin{proof}
We divide this proof into two cases based on the number of indices $i$ on which $\alpha_i(u)\neq \alpha_i(v)$. We note that there must be at least one such index, as $\alpha$ is a bijection. 

\begin{enumerate}
    \item Case 1: There is more than one index on which $\alpha_i(u)\neq \alpha_i(v)$. We know that there is no edge between $\alpha(u)$ and $\alpha(v)$ in this case, by definition of the Cartesian product. We also know that there are two $G_i'$ in which $u$ and $v$ are in different connected components, which is impossible if there is an edge between them, as that edge is a path between them in all but one $G_i'$. Thus, there is no edge between $u$ and $v$ either.
    
    \item Case 2: There is exactly one $k$ on which $\alpha_k(u)\neq \alpha_k(v)$. If there is no edge between $\alpha_k(u)$ and $\alpha_k(v)$ in $G_k^*$, then we know that there is no edge in $G_k'$ between $C_{\alpha_k(u)}$ and $C_{\alpha_k(v)}$ so we can't have an edge between $u$ and $v$ (as that would create such an edge). If there is any edge between $\alpha_k(u)$ and $\alpha_k(v)$ in $G_k^*$ of weight $w_{uv}$ then by Lemma \ref{claim:factorization1}, we have that all nodes in $C_{\alpha_k(u)}$ have an edge to exactly one node in $C_{\alpha_k(v)}$. Let $v'\in C_{\alpha_k(v)}$ be the node that $u$ has an edge to. Because they share an edge, $\alpha_i(v)=\alpha_i(u)=\alpha_i(v')$ for all $i\neq k$ so $\alpha(v)=\alpha(v')$. Since $\alpha$ is a bijection, this means $v=v'$ and thus $uv$ is an edge of the same weight as $\alpha(u)\alpha(v)$.
\end{enumerate}
The two parts of the proof together show the theorem.
\end{proof}

Theorem \ref{claim:factorization4} shows that there is an isomorphism between the input graph to Algorithm \ref{alg:general-alg} and the Cartesian product of the output graphs when the input relation is $(\theta\cup\tau)^*$, implying an $O(|E|^2)$ (plus APSP) time factorization of $G$.  (We will elaborate on this runtime in the following section.) We can also show that the output graphs themselves are prime. We do so using the following lemma.

\begin{lem}\label{claim:output-prime}
If $G$ is the input graph to Algorithm \ref{alg:general-alg} with $(\theta\cup\tau)^*$, the output graphs are all prime.
\end{lem}

\begin{proof}
Using Lemma \ref{claim:same-factor}, we only have to show that the edges in each factor are all related by $(\theta\cup\tau)^*$ (when evaluated on the factor itself). We know that a factor $G_i^*$ appears as an isometric subgraph of $\Pi_iG_i^*$ (and thus as an isometric subgraph of $G$) by definition of the Cartesian product. We know that all edges in this isometric subgraph are in the same equivalence class of $(\theta\cup\tau)^*$ because their pre-images under $\alpha$ are and $\alpha$ is an isomorphism. This implies their parent edges are all related by $(\theta\cup\tau)^*$ so $G_i^*$ is prime.
\end{proof}

Using the above lemma and theorem, we see that the algorithm in question produces a prime factorization of the input graph.

\subsection{Uniqueness of prime factorization}\label{sec:factorization-unique}
We additionally claim that for any graph $G$, there is at most one set $\{G_1,G_2,\ldots,G_m\}$ (up to graph isomorphisms) of graphs that form a prime factorization of $G$. 

\begin{thm}\label{claim:factorization-uniqueness}
For a given weighted graph $G$ and any two factorizations $\mathcal{G} = \{G_1, G_2, \dots, G_m\}$ and $\mathcal{G}' = \{G'_1, G'_2, \dots, G'_{m'}\}$ of $G$, $m=m'$ and there is a bijection $f : \mathcal{G} \to \mathcal{G}'$ such that $G_i$ is isomorphic to $G'_i$ for all $i, 1\le i \le m$.
\end{thm}

\begin{proof}
Let $\alpha$ and $\beta$ be isomorphisms from $G$ to the product of the graphs in $\mathcal{G}$ and $\mathcal{G'}$, respectively. We will define a mapping $f:\mathcal{G}\rightarrow \mathcal{G'}$ as follows. If there exists an edge 
$uv\in E(G)$ with $\alpha_i(u)\neq \alpha_i(v)$ and $\beta_j(u)\neq \beta_j(v)$, then $f(G_i)=G_j'$. 
We show that $f$ is a well-defined bijection. First, we note that by Lemma \ref{claim:factorization1}, two edges are related by $(\theta \cup \tau)^*$ if and only if they have parent edges in the same factor of any factorization. Thus, if there exists $uv,xy\in E(G)$ such that $uv$ and $xy$ have parent edges under $\alpha$ in factor $G_i$ of $\mathcal{G}$, then those edges must also have parent edges under $\beta$ in the same factor of $\mathcal{G'}$, meaning $G_i$ is mapped to exactly one $G_j'$. The reverse reasoning also shows that each graph in $\mathcal{G'}$ is mapped to by exactly one graph in $\mathcal{G}$. 

Now, we must show that for each $G_i\in \mathcal{G}$, $f(G_i)$ is isomorphic to $G_i$. We know that $G_i$ appears as a subgraph of $G$, with all edges related by $(\theta\cup\tau)^*$. Because this subgraph must appear in any Cartesian product and all edges must be related by this relation, we know that it must appear as a subgraph of $f(G_i)$. The reverse reasoning implies that $f(G_i)$ and $G_i$ are isomorphic.
\end{proof}

Thus, from this we get that Algorithm \ref{alg:general-alg} with input $(G,(\theta\cup\tau)^*)$ for a general graph $G$ produces a prime factorization of $G$.

\section{Computing the runtime of factorization and pseudofactorization}\label{sec:runtime-general}

In order to address runtime and prepare for the following section, we will discuss another view on how to compute the equivalence classes of the transitive closure of a  relation $R$ on the edges of a graph $G=(V,E,w)$ whose transitive closure is also symmetric and reflexive. To do this, we define a new unweighted graph $G_R=(V_R,E_R)$. The vertices of this graph are the edges of the original graph $G$ and there is an edge between two vertices if and only if the corresponding edges in the original graph are related by $R$. The connected components of the graph are then the equivalence classes of the original graph under the transitive closure of $ R$. The time to compute the connected components can be found using BFS in $O(|V_R|+|E_R|)$ time. As an upper bound, we know that once this graph is computed, $|V_R|=|E|$ and $|E_R|=O(|V_R|^2)=O(|E|^2)$. Thus, computing the connected components with BFS takes at most $O(|E|^2)$ time. If further bounds can be placed on the number of edges in $G_R$, this time can be decreased further. Figure \ref{fig:simple-relation-graph} shows an example of how $G_{\theta}$ is constructed for a given input graph.

\begin{figure}
    \centering
    \includegraphics[width=.8\linewidth]{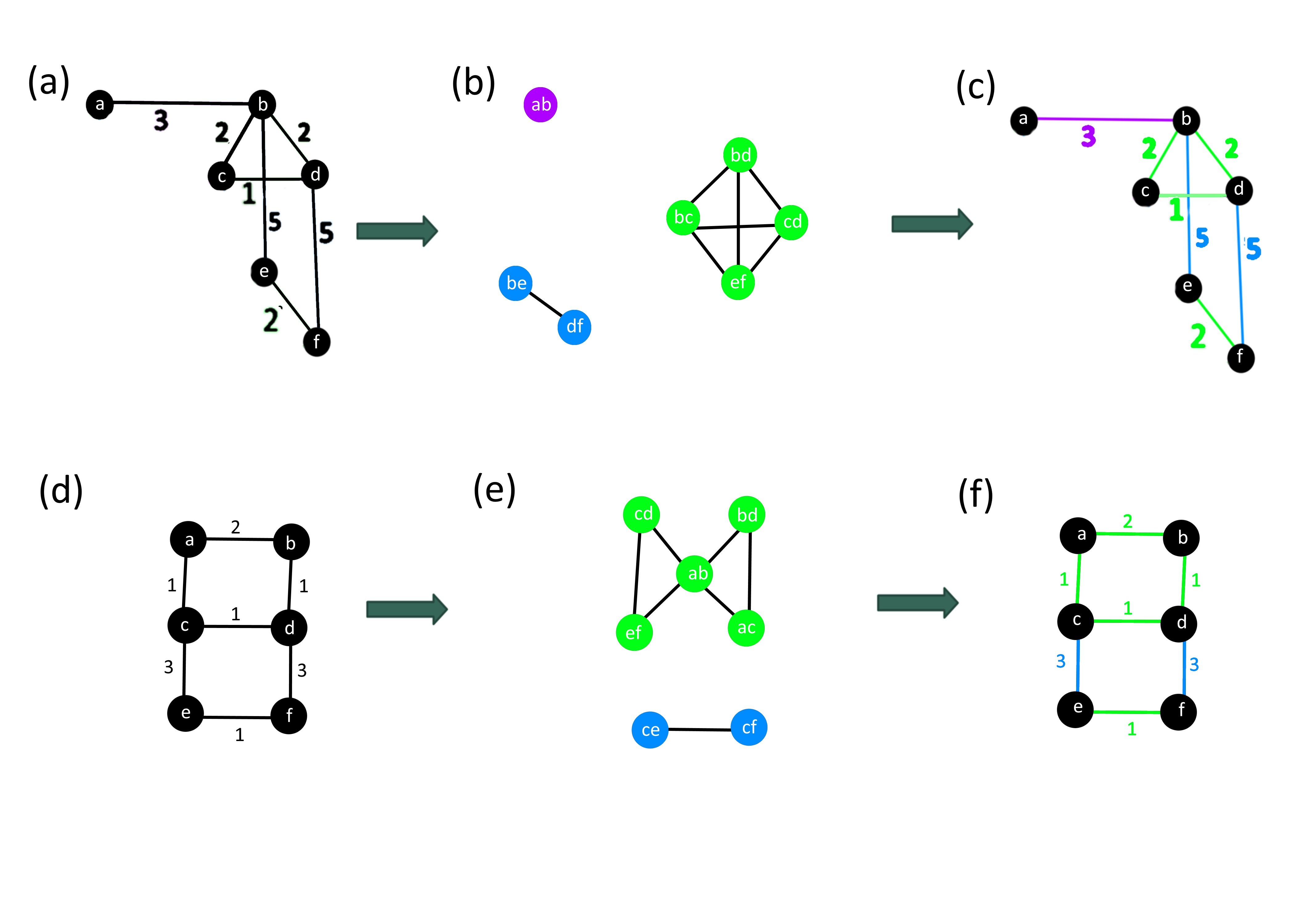}
    \caption[An example of how to construct $G_{\theta}$]{a) A weighted graph $G$. b) $G_\theta$. c) $G$ with edges colored according to their equivalence class under $\thetahrel$. Each equivalence class corresponds to a connected component of $G_\theta$. d-f) Analogous illustration for the graph shown in (d).}
    \label{fig:simple-relation-graph}
\end{figure}

\subsection{Runtime for pseudofactorization}\label{sec:pseudofactorization-runtime1}
To perform Graham and Winkler's algorithm, we compute the equivalence classes of $\hat \theta$, then for each equivalence class we perform linear time work by removing all edges in the equivalence class, computing the condensed graph, and checking which nodes in the new graph should have edges between them. (In this and all other computations in this paper, we don't actually have to check that all edges between connected components are equal, as that fact is guaranteed by the proofs in the previous sections.) In the worst case, each edge is in its own equivalence class, so we do $O(|V||E|+|E|^2)=O(|E|^2)$ work since the graph is connected. To compute the equivalence classes of $\hat \theta$, we compute $G_{\theta}$ and find the connected components. If we first find all pairs shortest path (APSP) distances, we can check if any pair of edges $ab,xy$ is related by $\theta$ in $O(1)$ time. Thus, once we have found all of these distances, we can compute all edges from a given node of $G_\theta$ in $O(|V_\theta|)=O(|E|)$ time, for a total of $O(|E|^2)$ time  and then we take $O(|E|^2)$ time to compute the connected components. Thus, total runtime is $O(|E|^2)$ plus APSP computation time, which is currently known to be $O(|V|^2\log\log |V| +|V||E|)$ \cite{apsp} and thus gives us a runtime of $O(|V|^2\log\log |V| + |E|^2)$ overall.

\subsection{Runtime for factorization\label{sec:factorization-runtime}}

For factorization using Algorithm \ref{alg:general-alg}, we can again bound the runtime of the main loop by $O(|E|^2)$ and thus just have to consider the time needed to compute the equivalence classes of $(\theta\cup\tau)^*$. We can first compute all distances using a known APSP algorithm. From here, we can determine all edges in $G_{\theta\cup \tau}$ that occur as a result of $\theta$ relations, and thus we only have to add in edges that occur as a result of $\tau$ relations. To do this, we can compare each pair of edges related by $\theta$. If we have two edges of the form $ab$ and $cd$ that are related by $\theta$ and have the same edge weight, we can check if $ac$ and $bd$ are edges and have equal weights and an edge between them in $G_\theta$. If so, we have that adjacent edges in this 4-cycle are not related by $\tau$. Compute the set $\mathcal{E}$ of all adjacent edges of $G$ not related by $\tau$, and then we can compute the set $\mathcal{E'}$ of all pairs of edges that are adjacent but not in $\mathcal{E}$. Add these edge pairs as edges in $G_\theta$ to compute $G_{\theta\cup\tau}$. For a given pair of edges in $G$ their contribution to $\mathcal{E}$ is found in constant time for $O(|E|^2)$ time to compute $\mathcal{E}$, and $\mathcal{E}'$ may also be found in $O(|E|^2)$ time given $\mathcal{E}$, so it is $O(|E|^2)$ time to construct the new graph. Thus, it is $O(|E|^2)$ total time to construct $G_{\theta\cup \tau}$ and another $O(|E|^2)$ to get its connected components. This brings our total runtime for graph factorization up to $O(|E|^2)$ plus APSP time, for $O(|V|^2\log\log |V|+|E|^2)$ overall.

\section{An algorithm for improved pseudofactorization runtime}\label{sec:pseudofactorization-runtime2}

In the section on pseudofactorization, we showed that Graham and Winkler's algorithm on a weighted graph can be used to pseudofactor a minimal weighted graph. For a weighted graph, this algorithm takes $O(|E|^2)$ time plus the time to compute all pairs shortest paths, which is currently bounded at $O(|V|^2\log \log |V|+|V||E|)$ \cite{apsp}. Thus, this algorithm takes $O(|V|^2\log\log |V| +|E|^2)$ time. For unweighted graphs, this time is just $O(|E|^2)$. Feder \cite{feder} showed that for unweighted graphs, rather than using the equivalence classes of $\hat\theta $, the same algorithm could use an alternative equivalence relation $\hat\theta_T$, which has the same equivalence classes but whose classes are faster to compute. Given a spanning tree $T$ of the graph, two edges $ab,xy$ are related by $\theta_T$ if and only if $ab \ \theta \ xy$ and at least one of $ab,xy$ is in the spanning tree $T$. As before, $\hat\theta_T$ is the transitive closure of $\theta_T$. Feder showed that for an arbitrary tree $T$, this equivalence relation could be used in Graham and Winkler's algorithm to produce a psuedofactorization of an unweighted input graph. Because $G_{\theta_T}$ can be computed in $O(|V||E|)$ time and $\hat\theta_T$ can only have up to $|V|-1$ equivalence classes, this brings the total time for computing the pseudofactorization of an unweighted graph down to $O(|V||E|)$. 

In the case of weighted graphs, we will show that we can find a tree $T^*$ such that $\hat \theta_{T^*}$ has the same equivalence classes as $\hat \theta$. First, we reiterate proof that $\hat\theta_{T^*}$ is an equivalence relation by showing that $\hat\theta_T$ is an equivalence relation on the edges of $G$ for any spanning tree $T$ of $G$. (This is no different than for unweighted graphs, a proof of which can be found in \cite{feder}, but for completeness we reiterate it here for weighted graphs.) First, since $\hat\theta_T$ is by definition transitive, we only need to show reflexivity and symmetry. The relation is clearly symmetric because $\theta$ is symmetric. Take $ab\in E(G)$. Because $T$ is a spanning tree, we know there is a path $P=(a=p_0,p_1,\ldots,p_n=b)$ from $a$ to $b$ that uses only edges in $T$. We get that $\sum_{i}[d(a,p_i)-d(a,p_{i-1})]-[d(b,p_i)-d(b,p_{i-1})]=[d(a,b)-d(a,a)]-[d(b,b)-d(b,a)]=2d(a,b)\neq 0$, so there must exist $p_jp_{j-1}$ such that $ab\thetarel p_jp_{j-1}$, and since $p_jp_{j-1}\in T$, this means that $ab\ \theta_T \ p_jp_{j-1} \ \theta_T \ ab$, where the last step is by symmetry. This means that $ab\ \hat\theta_T \ ab$ and thus the relation is reflexive as well. Thus, $\hat\theta_T$ is an equivalence relation for any $T$. We make the following claim  about the equivalence classes:

\begin{claim}\label{claim:equiv-classes-as-subsets}
Let $G$ be a graph with a spanning tree $T$ and let $uv\in E(G)$ be such that $[uv]_{\hat\theta_T}$ is the equivalence class under $\hat\theta_T$ with $uv$ and $[uv]_{\hat\theta}$ is the equivalence class under $\hat\theta$ with $uv$. Then $[uv]_{\hat\theta_T}\subseteq [uv]_{\hat\theta}$.
\end{claim}

\begin{proof}
Take $xy\in [uv]_{\hat\theta_T}$. Since $uv\ \hat\theta_T \ xy$, there must exist a sequence of edges such that \\ $uv=u_0v_0 \thetarel_T u_1v_1\thetarel_T \cdots \thetarel_T u_nv_n=xy$. We have that $u_iv_i\thetarel_T u_{i+1}v_{i+1}$ if and only if $u_iv_i \thetarel u_{i+1}v_{i+1}$ and at least one edge is in the tree. Thus, we know that $uv=u_0v_0 \thetarel u_1v_1\thetarel \cdots \thetarel u_nv_n=xy$, which means $uv\ \hat \theta \ xy$. Thus $xy\in [uv]_{\hat\theta}$.
\end{proof}

In fact, while we have not proven that $[uv]_{\hat\theta_{T}} = [uv]_{\hat\theta}$ for any $T$, we can prove that there is some tree $T^*$ such that $[uv]_{\hat\theta_{T^*}} = [uv]_{\hat\theta}$ for all $uv \in E(G)$. Algorithm \ref{alg:faster-alg} finds such a $T^*$, and once $T^*$ and APSP distances are found we can compute the equivalence classes of $\hat\theta_{T^*}$ in $O(|V||E|)$ time, as this only requires comparing each edge in $T^*$ to each edge in the graph and computing whether the two edges are related by $\theta$, which can be done in constant time for each pair.

Informally, Algorithm \ref{alg:faster-alg} starts with any spanning tree and repeatedly modifies it until the equivalence classes under $\hat\theta_T$ equal those under $\hat\theta$. It does so by looking at a particular equivalence class for $\hat\theta_T$ and for each edge in that equivalence class, checking that every edge on the simple path between its endpoints through the tree is in the current equivalence class or an already processed one, and swapping the edge in question into the tree if this does not hold. The idea behind this process is to try to grow the equivalence class we're working on as much as possible until it is the same as the equivalence class under $\hat\theta$. 

\begin{algorithm}
\caption{Algorithm for breaking up a graph over a relation}
\label{alg:faster-alg}
\textbf{Input}: A weighted graph $G$ and all pairs of distances between nodes. 

\textbf{Output}: A tree $T^*$ such that the equivalence classes of $\hat\theta$ are the same as those of $\hat\theta_{T^*}$ on $G$.

\begin{algorithmic}
    \State  Using BFS or DFS, find any spanning tree of $G$, which we will call $T$.
    \State Compute $G_{\theta_T}$
    \State Set $undiscovered \leftarrow V(G_{\theta_T})$
    \State Set \textit{current-class} $\leftarrow 0$
    \While{$undiscovered$ is not empty}
        \State Pick $xy\in undiscovered$
        \State Run BFS to find all nodes reachable from $xy$ in $G_{\theta_T}$ 
        \State Mark all newly discovered edges of $G$ with the number \textit{current-class}
        \State Set $reachable\leftarrow$ all nodes reachable from $xy$ in $G_{\theta_T}$
        \While{$reachable$ is not empty} 
            \State Pick $ab\in reachable$
            \State Find the path $P$ from $a$ to $b$ in the tree $T$
            \If{there is an edge $uv\in P$ that is not marked with a class number} 
                \State Create a new spanning tree $T'$ from $T$ by adding $ab$ and removing $uv$
                \State Set $T\leftarrow T'$
                \State Update $G_{\theta_T}$ based on the new $T$ 
                \State Update $reachable$ to add any newly reachable nodes from $xy$ in $G_{\theta_T}$
                \State Mark any newly reachable edges with \textit{current-class}
            \EndIf 
            \State Remove $ab$ from $reachable$ and from $undiscovered$
        \EndWhile
        
        \State \textit{current-class} $\leftarrow $ \textit{current-class} $+1$
    \EndWhile
    \State return $T$
            

\end{algorithmic}
\end{algorithm}

We will first justify that this algorithm works correctly and then that it runs in $O(|V||E|)$ time. In particular, during the runtime section, we will go into more detail about how to update $G_{\theta_T}$ efficiently, but for now we take for granted that we update the graph correctly whenever needed.

\subsection{Correctness of Algorithm \ref{alg:faster-alg}}

To prove that Algorithm \ref{alg:faster-alg} works correctly, we will show an invariant for the outer while loop.

\textbf{Invariant}: At the top of the while loop, if an edge $xy$ is not in the set $undiscovered$, then for the current tree $T$ $[xy]_{\hat\theta_T}=[xy]_{\hat\theta}$ where $[xy]_R$ is the equivalence class of $xy$ under the equivalence relation $R$.

This invariant clearly holds in the base case, as in the first round of the while loop, all edges are in $undiscovered$ and thus this is vacuously true. 

Note that $[xy]_{\hat\theta_T}=[xy]_{\hat\theta}$ if and only if $xy$'s connected component in $G_{\theta_T}$ includes the same nodes as that in $G_\theta$. Thus, another way of  phrasing the invariant is by saying that for any $xy\in V(G_{\theta_T})$ such that $xy\notin undiscovered$, $C_{\theta_T}(xy)$ (the connected component in $G_{\theta_T}$ containing $xy$) has the same nodes as $C_\theta(xy)$ (the connected component in $G_\theta$ containing $xy$).

Consider a particular iteration of the while loop in which the edge we select from $undiscovered$ is $xy$. Throughout the loop, we add and remove some edges from $G_{\theta_T}$ as we alter the tree. In particular, when we remove an edge $uv$ from $T$, we may remove some edges that are incident to the node $uv$ in $G_{\theta_T}$. Removing edges from the tree $T$ to create a new tree $T'$ can potentially cause two edges, $u_1v_1$ and $u_2v_2$ that were related by $\theta_T$ not to be related by $\theta_{T'}$, but in the following lemma we restrict the kind of edges for which this may be true.

\begin{lem}\label{claim:faster1}
If an edge $uv$ is removed from a spanning tree $T$ and replaced with a new edge $ab$ to create a new spanning tree $T'$ and we have edges $u_1v_1,u_2v_2$ such that $u_1v_1\ \theta_T \ u_2v_2$ but $u_1v_1\ \cancel\theta_{T'} \ u_2v_2$, then one of $u_1v_1$ and $u_2v_2$ is equal to $uv$.
\end{lem}

\begin{proof}
We know that $G_{\theta_T}$ and $G_{\theta_{T'}}$ are identical to $G_{\theta}$, but with edge set restricted to those edges incident to nodes that correspond to edges in $T$ and $T'$, respectively. Thus, by swapping out $uv$ for $ab$, the only edges that may have been deleted from $G_{\theta_T}$ to form $G_{\theta_{T'}}$ are those incident to $uv$. Since two nodes are adjacent in $G_{\theta_{T}}$ iff they correspond to two edges of $G$ that are related by $\theta_T$, this means that the only way two edges can be related by $\theta_T$ but not by $\theta_{T'}$ is if one of them is $uv$.
\end{proof}

\begin{figure}
    \centering
    \includegraphics[width=\linewidth]{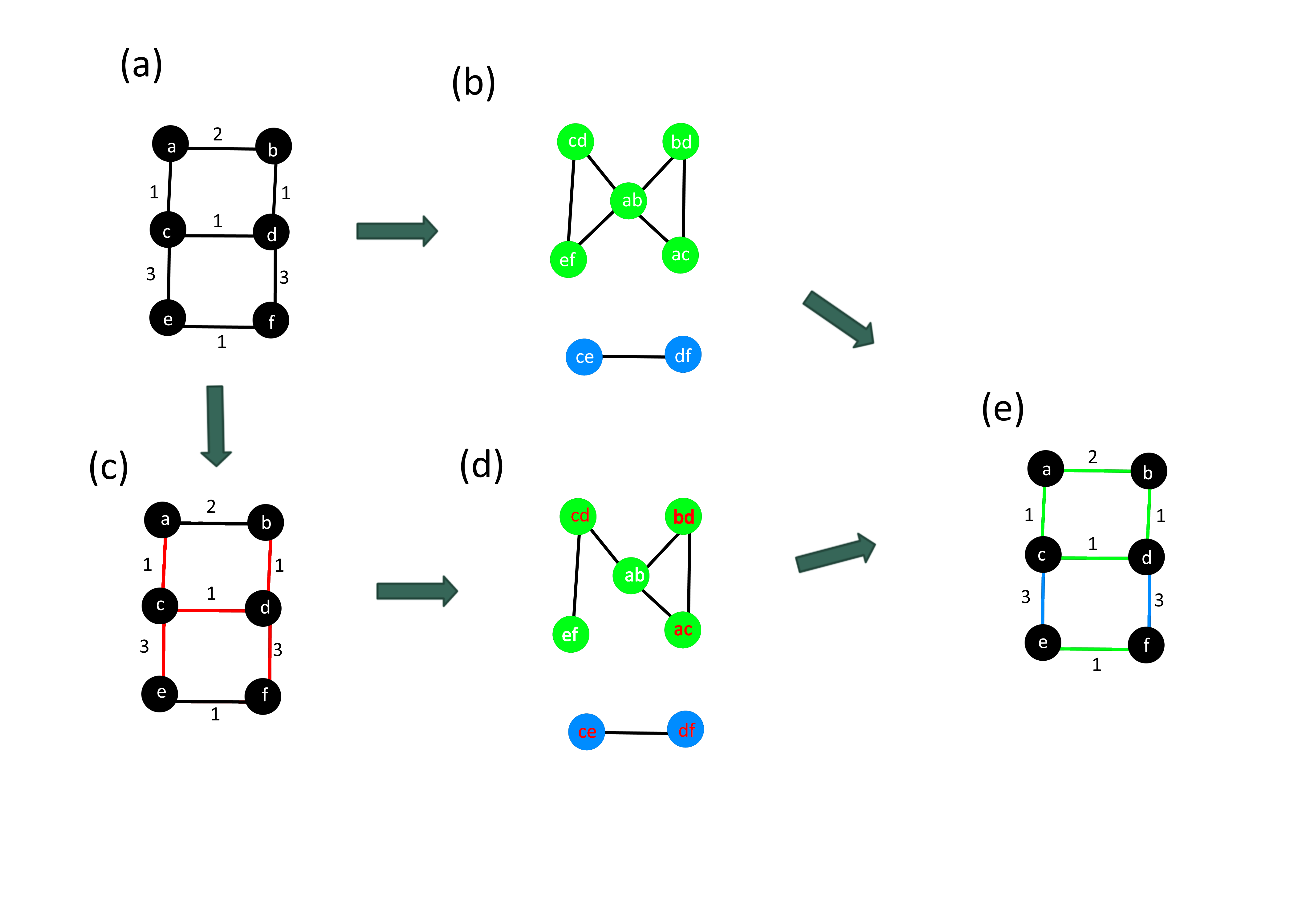}
    \caption[An visualization of the relationship between $G_\theta$ and $G_{\theta_T}$]{(a) An input graph $G$. (b) $G_\theta$. (c) $G$ with a spanning tree $T$ (edges highlighted in red). (d) $G_{\theta_T}$. The nodes with names in red correspond to edges in $G$ the spanning tree $T$. The graph is identical to $G_\theta$ but with edges not incident to a red node removed. (e) For this $T$, $G_\theta$ and $G_{\theta_T}$ produce the same equivalence classes on the edges of $G$.}
    \label{fig:tree-relation}
\end{figure}

To follow up on the proof of Lemma \ref{claim:faster1}, we can introduce a new picture of what $G_{\theta_T}$ looks like relative to $G_\theta$, represented in Figure \ref{fig:tree-relation}. In particular, if we take $G_\theta$ and color red all the nodes corresponding to edges in $T$, deleting the edges not incident to at least one red node produces $G_{\theta_T}$. This helps provide a visual representation for the relationship between $\hat\theta$ and $\hat\theta_T$.


By the invariant, at the beginning of a particular iteration of the while loop, we assume that for all $uv\notin undiscovered$, we have $C_{\theta_T}(uv)=C_\theta(uv)$. In the following lemma, we assert that at the end of the given iteration of the while loop, this will remain true for all edges that were discovered prior to this iteration.

\begin{lem}\label{claim:faster2}
If $st$ is a node in $G_{\theta_T}$ that has been discovered when the spanning tree $T$ is updated to a new tree $T'$ by removing an edge $uv$ and adding an edge $ab$, then $C_{\theta_T}(st)\subseteq C_{\theta_{T'}}(st)$
\end{lem}
\begin{proof}
By Lemma \ref{claim:faster1}, when this update is done, the only edges that are removed from $G_{\theta_T}$ to create $G_{\theta_{T'}}$ are those incident to the node $uv$. However, since $st$ is discovered at this point, we know that all nodes reachable from $st$ have also been discovered, as the algorithm updates knowledge about explored connected components every time it updates the graph. Since we never remove discovered edges from the tree, we know that $uv$ is undiscovered, and thus the edges we removed in $G_{\theta_T}$ are not adjacent to anything in $st$'s connected component, meaning that $st$'s connected component in $G_{\theta_T}$ is a subset of that in $G_{\theta_{T'}}$ and we have $C_{\theta_T}(st)\subseteq C_{\theta_{T'}}(st)$.
\end{proof}



Given Lemma \ref{claim:faster2}, in order to show the invariant, we only have to show that if $xy$ is chosen at the beginning of an iteration of the while loop, at the end of that iteration of the while loop, when we have a tree $T$, $C_{\theta_T}(xy)=C_\theta(xy)$.
Thus, our current goal rests on showing that for each $xy$ discovered at the beginning of an iteration of the outer while loop, if $T$ is the tree at the end of that iteration of the while loop, $[xy]_{\hat\theta_T}=[xy]_{\hat\theta}$.  We begin by considering an edge as ``processed'' after we have first discovered it and examined the path between its endpoints in the tree. We are able to make the following observation about the path between the endpoints of each processed edge.

\begin{lem}\label{claim:faster4}
If tree $T$ is updated to tree $T'$ by removing an edge $uv$ from the graph and adding a new edge $ab$ to the graph, then for any processed edge $a'b'$, the path from $a'$ to $b'$ through $T'$ consists only of edges that we have already discovered/marked in the $\theta$ graph.
\end{lem}
\begin{proof}
We know that immediately after processing an edge $a'b'$, we updated the tree such that this lemma held. Since we never remove marked/discovered edges from the tree, this means that this path from $a'$ to $b'$ consisting of only marked edges still exists in the tree and since there is only one path from $a'$ to $b'$ through the tree, the lemma holds.
\end{proof}

Finally, using this lemma we are able to reach our final conclusion about $xy$'s equivalence class, which by our earlier analysis tells us that at the end of the algorithm, the equivalence classes of $\hat\theta_{T^*}$ are those of $\hat\theta$, as desired.

\begin{lem}\label{claim:faster5}
If $xy$ is an edge discovered at the beginning of a particular iteration of the outer while loop and $T$ is the tree at the end of that iteration of the while loop, then $[xy]_{\hat\theta_T}=[xy]_{\hat\theta}$.
\end{lem}

\begin{proof}
From Lemma \ref{claim:faster4} and the fact that we process every edge in $xy$'s connected component of $G_{\theta_T}$, we know that at the end an iteration of the outer while loop, all edges in $C_{\theta_T}(xy)$ have paths through $T$ that include only marked edges, which are edges that are in this connected component or some previously processed connected component in $G_{\theta_T}$. Now, assume $[xy]_{\hat\theta_T}\neq[xy]_{\hat\theta}$. Since we know that $[xy]_{\hat\theta_T}\subseteq[xy]_{\hat\theta}$, this means that $[xy]_{\hat\theta_T}\subset[xy]_{\hat\theta}$. Pick $uv\in [xy]_{\hat\theta},uv\notin [xy]_{\hat\theta_T}$. We know that because they're in the same $\hat\theta$ equivalence class, there is a sequence of edges such that $xy=u_1v_1\thetarel u_2v_2\thetarel \cdots \thetarel u_kv_k=uv$. If $u_iv_i\ \hat\theta_T \ u_{i+1}v_{i+1}$ for each $i$, then by transitivity we would have $xy \ \hat\theta_T \ uv$ and thus they'd be in the same equivalence class under $\hat\theta_T$. This means we can assume there exists $j$ such that $u_jv_j$ and $u_{j+1}v_{j+1}$ are not related by $\hat\theta_T$ but $u_jv_j\in [xy]_{\hat\theta_T}$. Pick the first such $j$ and to simplify notation, we will call these two edges $ab$ and $st$ respectively. We have $ab\thetarel st$, but $ab\in [xy]_{\hat\theta_T}$ and $st\notin [xy]_{\hat\theta_T}$.

Consider the path $Q=(a=q_0,q_1,\ldots,q_n=b)$ from $a$ to $b$ through $T$. We get:
\begin{align*}
    \sum_i[d(s,q_i)-d(s,q_{i+1})]-[d(t,q_i)-d(t,q_{i+1})] &= [d(s,a)-d(s,b)]-[d(t,a)-d(t,b)] \\
    & \neq 0,
\end{align*}
where the equality is by telescoping and the inequality is by the fact that $ab\thetarel st$. We know that this means $st \thetarel s^*t^*$ for some $s^*t^*\in Q$. Since $s^*t^*\in T$ and $st\thetarel s^*t^*$, we get $s^*t^*\in [st]_{\hat\theta_T}$. Thus, on this path there exists $s^*t^*\in [st]_{\hat\theta_T}\neq [xy]_{\hat\theta_T}$.

By our earlier lemma, we know that $s^*t^*$ is marked and that every marked edge is either in $xy$'s equivalence class or an equivalence class processed on a previous iteration of the while loop. Since we know $s^*t^*$ is not in $xy$'s equivalence class, it must be in an equivalence class we processed on a previous iteration of the while loop. However, our invariant tells us that this means $[s^*t^*]_{\hat \theta_T}=[s^*t^*]_{\hat \theta}=[xy]_{\hat \theta}$ (since $xy\thetarel s^*t^*$). However, since this equivalence class did not change over the course of this iteration of the while loop, this means that $xy$ was in an equivalence class we already processed before this iteration and thus it was marked, so it was not in $undiscovered$, a contradiction. 

Thus, we know that at the end of the round, every edge not in $undiscovered$ has $[xy]_{\hat\theta_T}=[xy]_{\hat \theta}$. Since we remove at least one edge from $undiscovered$ in each iteration of the while loop, the loop terminates with everything popped and we get that all equivalence classes of $\hat\theta_T$ are the same as those of $\hat\theta$. Thus, we have shown that if the invariant holds at the beginning of an iteration of the outer while loop, it holds at the end.
\end{proof}

This lets us know that Algorithm \ref{alg:general-alg} on the input $(G,\hat\theta_{T^*})$ where $T^*$ is the output of Algorithm \ref{alg:faster-alg} on $G$, produces an irreducible pseudofactorization of $G$. From here, we analyze the runtime of Algorithm \ref{alg:faster-alg} to determine if using it to find a new equivalence relation on the edges of the graph actually improves our runtime.

\subsection{Runtime of Algorithm \ref{alg:faster-alg}}

For a graph $G=(V,E,w_G)$, finding some spanning tree $T$ takes $O(|V|+|E|)$ time. When we compute $G_{\theta_T}$, we know that every edge must be adjacent to some $uv\in T$, as if two edges in $G$ are related by $\theta_T$, then at least one of them must be in $T$. Thus, if we already have the APSP distances and can thus check if two edges are related by $\theta$ in constant time, it takes us $O(|V||E|)$ time to construct $G_{\theta_T}$, as we compare each edge in $T$ to every other edge in order to compute the edges of $G_{\theta_T}$.

We notice that each time we look for the path from $a$ to $b$ in the tree, we pop $ab$ from $unreachable$, and we never do this for a node not in $unreachable$, so we do this once for each edge. As it takes $O(|V|)$ time to find the path from $a$ to $b$ in the tree and we do it for $|E|$ edges, this is $O(|V||E|)$ total time for examining the paths in the tree. 
This leaves us just considering the total BFS time for all of the updates.

Consider a single iteration of the outer while loop, in which we pop edge $xy$. We begin by discovering all edges that $xy$ can reach, each of which is popped from $undiscovered$. At this point, if we add a new edge $ab$ to the tree, we need to update $reachable$, as we have potentially added new edges adjacent to $ab$ and thus potentially added new nodes to $xy$'s connected component. At this point, we can resume BFS beginning at $ab$, visiting only nodes that we have not seen before. We note that this means we may consider the edges adjacent to $ab$ a second time, but all other edges we traverse on this BFS run are new edges, as the origin of each edge must be a newly discovered node. 

When we add a new edge $ab$ into the tree, we must determine all edges out of $ab$ in the new $G_{\theta_T}$, which requires checking if $ab$ is related by $\theta$ to all other edges in the graph. We also remove an edge from the tree, which means we must remove all of its edges to other edges not in the tree in $G_{\theta_T}$. Each of these updates takes $O(|E|)$ time, so we must limit the number of times we do this. We note that no edge discovered in $G_{\theta_T}$ is ever removed from the graph, and at the point an edge $ab$ is added to the tree, it has been discovered in $G_{\theta_T}$ and thus will never be removed. This means that we can add at most $|V|-1$ edges to the tree. This means the updates to $G_{\theta_T}$ take at most $O(|V||E|)$ time total across all updates.

We know from our earlier analysis that when a new edge is added to the tree, we end up re-traversing at most $O(|E|)$ edges, and since this happens at most $|V|-1$ times, re-traversing edges in $G_{\theta_T}$ happens at most $O(|V||E|)$ times. Aside from that, we do a full BFS search of the graph $G_{\theta_T}$, which does have edges added and removed as we go along. However, we note that it has only $O(|V||E|)$ edges at the beginning and as we have just argued, only $O(|V||E|)$ new edges are added, so even if we manage to traverse every edge that appeared in any iteration of the graph (which actually doesn't happen since removed edges are connected to undiscovered nodes, which means they haven't been traversed), we take only $O(|V||E|)$ time to do so.

Thus, the total time for this algorithm is $O(|V||E|)$. We also note that because every equivalence class must contain an edge from $T^*$, there are most $|V|-1$ equivalence classes for this relation. Thus, using this relation with Algorithm \ref{alg:general-alg}, the primary \textit{for} loop is only traversed $O(|V|)$ times. Since Algorithm \ref{alg:faster-alg} both provides a tree to $T^*$ to use for $\hat \theta_{T^*}$ and computes the relation's equivalence classes in $O(|V||E|)$ time given the APSP distances, total time for pseudofactoring a graph using the equivalence relation $\hat\theta_T$ for $T$ produced by Algorithm \ref{alg:faster-alg} takes only $O(|V||E|)$ plus APSP time, as opposed to $O(|E|^2)$ plus APSP time for the same algorithm with $\hat\theta$ as the input relation. Using current APSP algorithms, this is a speedup from $O(|V|^2\log\log |V| + |E|^2)$ to $O(|V|^2\log\log |V|+|V||E|)$ for overall runtime.



\section{Conclusion}\label{sec:conclusion}

In areas like molecular engineering, distance metrics are complex and frequently not representable in unweighted graphs. Isometric embeddings of weighted graphs, while more difficult than for unweighted graphs, open up new avenues for solving problems in these areas. Here, we extended factorization and pseudofactorization to minimal weighted graphs and provided polynomial-time algorithms for computing both, in $O(|V|^2\log\log |V| + |E|^2)$ time and $O(|V|^2\log\log|V|+|V||E|)$ time, respectively. Several open questions remain, including the following:

\begin{itemize}
    \item Can the efficiency of weighted graph factorization be improved, as weighted graph pseudofactorization was in Section \ref{sec:pseudofactorization-runtime2}, to $O(|V||E|)$ when distances are precomputed?
    \item Can we place lower bounds on the time needed to find a graph's prime factorization or irreducible pseudofactorization? In particular, can we lower bound these processes by $O(|V||E|)$ time?
\end{itemize}

In partial response to the first question, we can see by the proofs given in Section 7 that a modification of Algorithm \ref{alg:faster-alg} can be used to find a tree $T$ such that $(\theta_T\cup \tau)^*$ has the same equivalence classes as $(\theta\cup \tau)^*$, but at the moment it is unclear if the time determining all pairs of edges related by $\tau$ can be bounded or if $\tau$ can be modified in a way that makes these relations faster to compute.



\section*{Acknowledgements}
M.B. and J.B. were supported by the Office of Naval Research (N00014-21-1-4013), the Army Research Office (ICB Subaward KK1954), and the National Science Foundation (CBET-1729397 and OAC-1940231). M.B., J.B., and K.S. were supported by the National Science Foundation (CCF-1956054). Additional funding to J.B. was provided through a National Science Foundation Graduate Research Fellowship (grant no. 1122374). A.C. was supported by a Natural Sciences and Engineering Research Council of Canada Discovery Grant. V.V.W. was supported by the National Science Foundation (CAREER Award 1651838 and grants CCF-2129139 and CCF-1909429), the Binational Science Foundation (BSF:2012338), a Google Research Fellowship, and a Sloan Research Fellowship.

\section*{Declarations of Interest}
The authors declare that they have no competing interests.


\bibliographystyle{plain} 
\bibliography{main.bib}





\end{document}